\documentclass[a4paper,UKenglish]{lipics-v2016}

\usepackage{amsmath}
\usepackage{amsthm}
\usepackage{amssymb}
\usepackage{listings}
\usepackage{tikz}
\usepackage{complexity}

\lstdefinelanguage{ML}{
  alsoletter={*},
  morekeywords={datatype, of, if, *},
  sensitive=true,
  morecomment=[s]{/*}{*/},
  morestring=[b]"
}

\lstdefinelanguage{scala}{
  alsoletter={@=>},
  morekeywords={nothing, abstract, case, catch, choose, class, def, do, else, extends, final, finally, for, if, implicit, import, match, new, null, object, let,
override, package, private, protected, requires, return, sealed, super, this, throw, trait, try, type, val, var, while, yield, domain, template, res, time,
postcondition, precondition,invariant, constraint, assert, each, _, return, @generator, ensure, require, ensuring, assuming, otherwise, asserting}
  sensitive=true,
  morecomment=[l]{//},
  morecomment=[s]{/*}{*/},
  morestring=[b]"
}

\newcommand{\codestyle}{\small\sffamily}

 \usepackage{array}
\usepackage{todonotes}
\usepackage{float}
\usepackage{stmaryrd}
\usepackage{hyperref}
\usepackage{algorithm}
\usepackage{thmtools}
\usepackage{thm-restate}
\usepackage{tikz}
\usetikzlibrary{automata, arrows, calc}

\usepackage{microtype}%

\declaretheorem[name=Theorem,numberwithin=subsection]{thm}
\declaretheorem[name=Lemma,numberwithin=subsection]{lem}

\begin{document}

\newcommand{\cL}{{\cal L}}
\newcommand*\circled[1]{\tikz[baseline=(char.base)]{
		\node[shape=circle,draw,inner sep=2pt] (char) {#1};}}

\newcommand{\circq}{\mbox{$\bigcirc$\kern-0.74em{\texttt{?}}\hspace{0.3em}}}

\newcommand{\method}[1]{\texttt{\textbf{#1}}}

\newcommand{\symbolstring}{\ensuremath{\mathrm{string}}}
\newcommand{\symbolint}{\ensuremath{\mathrm{int}}}
\newcommand{\symbolboolean}{\ensuremath{\mathrm{bool}}}

\newcommand\pto{\mathrel{\ooalign{\hfil$\mapstochar$\hfil\cr$\to$\cr}}}
\newcommand{\parentof}{\ensuremath{\rightarrow}}
\newcommand{\ancestorof}{\ensuremath{\rightsquigarrow}}
\newcommand{\parentofconcrete}{\ensuremath{\stackrel{c}{\rightarrow}}}
\newcommand{\ancestorofconcrete}{\ensuremath{\stackrel{c}{\rightsquigarrow}}}

\newcommand{\assignment}{\morph}

\newcommand{\nonterm}[1]{\ensuremath{\texttt{N}_{\texttt{#1}}}}
\newcommand{\term}[2]{\ensuremath{\texttt{T}_{\texttt{#1},\{#2\}}}}
\newcommand{\termX}[1]{\ensuremath{x_{#1}}}

\newcommand{\grule}[2]{\ensuremath{\nonterm{#1} \rightarrow #2}}

\newcommand*{\ruleset}[1]{%
	\begin{center}\begin{tabular}{lcl}%
			\ruleScan#1\relax\relax
		}
		\newcommand{\ruleScan}[2]{%
			\ifx\relax#1%
		\end{tabular}\end{center}%
		\else
		\nonterm{#1} & $\rightarrow$ & #2\\
		\expandafter\ruleScan
		\fi
	}
\newcommand{\rulesep}{\ensuremath{\vert}~}

\newcounter{numberofttodos}

\newcommand{\set}[1]{\{{#1}\}}
\newcommand{\ab}{\allowbreak}
\newcommand{\vtodo}[1]{\todo[size=\tiny, color=green!40]{#1}}
\newcommand{\ttodo}[1]{
    \stepcounter{numberofttodos}
    \todo[inline, color=blue!40]{
        (\thenumberofttodos{}) #1
    }
}
\newcommand{\ttodotext}[1]{
    \stepcounter{numberofttodos}
    \todo[inline, color=yellow!40]{
        (\thenumberofttodos{}) \colorbox{yellow!40}{[TEXT]} #1
    }
}
\newcommand{\ttodoproof}[1]{
    \stepcounter{numberofttodos}
    \todo[inline, color=red!40]{
        (\thenumberofttodos{}) \colorbox{red!40}{[PROOF]} #1
    }
}
\newcommand{\ttodoexample}[1]{
    \stepcounter{numberofttodos}
    \todo[inline, color=green!40]{
        (\thenumberofttodos{}) \colorbox{green!40}{[EXAMPLE]} #1
    }
}
\newcommand{\ttodoplan}[1]{
    \stepcounter{numberofttodos}
    \todo[inline, color=blue!40]{
        (\thenumberofttodos{}) \colorbox{blue!40}{[PLAN]} #1
    }
}

\newcommand{\adt}{A}

\newcommand{\yes}{{\sf Yes}}
\newcommand{\no}{{\sf No}}
\newcommand{\solutions}{{\tt sol}}

\newcommand{\treetestset}{tree test set}
\newcommand{\shape}{shape}
\newcommand{\regular}{sequential}

\newcommand{\Tree}{\mathcal{T}}
\newcommand{\TreeA}[1]{\Tree_{#1}}
\newcommand{\Nat}{\mathbb{N}}
\newcommand{\Var}{\mathbb{X}}

\newcommand{\sample}{\mathcal{S}}
\newcommand{\morph}{\mu}
\newcommand{\eq}{e}
\newcommand{\form}{\varphi}
\newcommand{\concat}{y}

\newcommand{\gram}{G}
\newcommand{\NTerm}{N}
\newcommand{\Term}{\Sigma}
\newcommand{\Prod}{R}
\newcommand{\Start}{S}
\newcommand{\getyield}[1]{{\sf yield}({#1})}
\newcommand{\rhs}{\mathit{rhs}}
\newcommand{\nterm}{A}
\newcommand{\mkrule}[2]{{#1} \rightarrow {#2}}
\newcommand{\rul}{r}
\newcommand{\west}[1]{\pi({#1})}
\newcommand{\east}[1]{\overline{\pi}({#1})}
\newcommand{\closing}[1]{\overline{#1}}

\newcommand{\final}{\bot}
\newcommand{\getgraph}{{\sf graph}}
\newcommand{\gettree}{{\sf tree}}
\newcommand{\getpath}{{\sf path}}
\newcommand{\chemin}{P}
\newcommand{\optimalset}{\Phi}

\newcommand{\rp}{y}
\newcommand{\spec}{\#}
\newcommand{\mkstate}[3]{{#1}_{({#2},{#3})}}
\newcommand{\substring}[3]{{#1}[#2,#2+#3]}

\newcommand{\stw}{{\sf 1STS}}
\newcommand{\trans}{\tau}
\newcommand{\utrans}{\tau_u}
\newcommand{\State}{Q}
\newcommand{\Alphabet}{\Sigma}
\newcommand{\initial}{\mathit{init}}
\newcommand{\Output}{\Gamma}
\newcommand{\sh}{S}
\newcommand{\semantics}[1]{{\llbracket}{#1}{\rrbracket}}
\newcommand{\domain}{D}
\newcommand{\default}[1]{\trans_{#1}}
\newcommand{\getmorph}[1]{{\tt morph}[{#1}]}
\newcommand{\gettransducer}[1]{{\tt sts}({#1})}

\newcommand{\ranked}[2]{{#1}^{({#2})}}
\newcommand{\noda}{A}
\newcommand{\nodb}{B}
\newcommand{\nodc}{C}
\newcommand{\leaf}{F}

\newcommand{\SigmaPairs}{\overline{\Sigma}}

\newcommand{\restrict}[2]{{#1}_{|{#2}}}
\newcommand{\project}[2]{{#1}_{|{#2}}}

\newcommand{\lcs}{{\sf lcs}}
\newcommand{\lcp}{{\sf lcp}}
\newcommand{\dom}{{\mathit dom}}
\newcommand{\getequation}{{\sf equation}}
\newcommand{\getrequation}{{\sf regEquation}}
\newcommand{\getroot}{{\sf root}}

\newcommand{\startsymbol}{\ensuremath{S_0}}
\newcommand{\Lin}[1]{\ensuremath{\text{Lin}(#1)}}
\newcommand{\AllInit}[1]{\ensuremath{\text{AllInit}(#1)}}
\newcommand{\testset}[1]{\ensuremath{T\!S_{#1}}}

\newcommand{\alts}{A}
\newcommand{\getautomaton}[2]{{\sf automaton}({#1},{#2})}
\newcommand{\tr}{\mathit{tr}}

\newcommand{\cl}{C}
\newcommand{\oneinthree}{{\sf OneInThree}}
\newcommand{\lf}{{\sf nil}}

\newcommand{\before}[1]{{\sf left}({#1})}
\newcommand{\after}[1]{{\sf right}({#1})}

\newcommand{\invoked}{T_{\textit{inv}}}
\newcommand{\winvoked}{W_{\textit{inv}}}
\newcommand{\tmin}{T_{\textit{min}}}

\newcommand{\exampledomain}{\ensuremath{D_{html}}}
\newcommand{\nnil}{\ensuremath{\text{nil}}}
\newcommand{\cons}{\ensuremath{\text{cons}}}
\newcommand{\node}{\ensuremath{\text{node}}}
\newcommand{\ddiv}{\ensuremath{\text{div}}}
\newcommand{\dpre}{\ensuremath{\text{pre}}}
\newcommand{\dspan}{\ensuremath{\text{span}}}
\newcommand{\listType}{\ensuremath{\text{List}}}
\newcommand{\tagType}{\ensuremath{\text{Tag}}}
\newcommand{\nodeType}{\ensuremath{\text{Node}}}

\newcommand{\userinput}[1]{{\sf #1}}
\newcommand{\getlang}[1]{{\mathcal L}({#1})}

\newcommand{\ourtitle}{Proactive Synthesis of Recursive Tree-to-String Functions from Examples}

\title{\ourtitle\footnote{
This work was partially supported by European Research Council (ERC) Project Implicit Programming and an EPFL-Inria Post-Doctoral grant.}%
}

 \author[1]{Mika\"el Mayer}
 \author[1]{Jad Hamza}
 \author[1]{Viktor Kun\v{c}ak}
 \affil[1]{EPFL IC IINFCOM LARA, INR 318, Station 14, CH-1015 Lausanne\\
   \texttt{firsname.lastname@epfl.ch}}
\authorrunning{Mika\"el Mayer, Jad Hamza and Viktor Kun\v{c}ak}
\Copyright{Mika\"el Mayer, Jad Hamza and Viktor Kun\v{c}ak}

\maketitle

\begin{abstract}

Synthesis from examples enables non-expert users to generate
programs by specifying examples of their behavior. A
domain-specific form of such synthesis has been recently
deployed in a widely used spreadsheet software product.  In
this paper we contribute to foundations of such techniques
and present a complete algorithm for synthesis of a class of
recursive functions defined by structural recursion over a
given algebraic data type definition.  The functions we consider map an
algebraic data type to a string; they are useful for, e.g.,
pretty printing and serialization of programs and data.  We
formalize our problem as learning deterministic sequential
top-down tree-to-string transducers with a single state (1STS).

The first problem we consider is learning a tree-to-string
transducer from any set of input/output examples
provided by the user. We show that, given a set of input/output
examples, checking whether there exists a 1STS consistent with 
these examples is NP-complete in general.
In contrast, the problem can be solved in polynomial time
under a (practically useful) closure condition that each
subtree of a tree in the input/output example set is also
part of the input/output examples.

Because coming up with relevant input/output examples may be
difficult for the user while creating hard constraint problems
for the synthesizer, we also study a more automated
active learning scenario in which the algorithm chooses the
inputs for which the user provides the outputs. Our
algorithm asks a worst-case linear number of queries as a
function of the size of the algebraic data type definition
to determine a unique transducer.

To construct our algorithms we present two new results on
formal languages.

First, we define a class of word equations, called
sequential word equations, for which we prove that
satisfiability can be solved in deterministic polynomial
time.  This is in contrast to the general word equations for
which the best known complexity upper bound is in linear space.

Second, we close a long-standing open problem about the
asymptotic size of test sets for context-free languages. A
test set of a language of words $L$ is a subset $T$ of $L$
such that any two word homomorphisms equivalent on $T$ are
also equivalent on $L$. We prove that it is possible to
build test sets of cubic size for context-free languages,
matching for the first time the lower bound found 20 years
ago.

\end{abstract}

\subjclass{F.3.1 Specifying and Verifying and Reasoning about Programs -- D.3.4 Processors}%
\keywords{programming by example, active learning, program synthesis}

\section{Introduction}

Synthesis by example has been very successful to help users deal 
with the tedious task of writing a program. This technique allows the user 
to specify input/output examples to describe the intended behavior of a
desired program. Synthesis will then inspect the examples given by the user, 
and generalize them into a program that respects these examples, and that 
is also able to handle other inputs.

Therefore, synthesis by example allows non-programmers to write programs 
without programming experience, and gives experienced users one more
way of programming that could fit their needs.
Current synthesis techniques usually rely on domain-specific 
heuristics to try and 
infer the desired program from the user. When there are multiple 
(non-equivalent) programs which are compatible with input/output 
examples provided by the user, these heuristics may fail to 
choose the program that the user had in mind when writing the examples.

We believe it is important to have algorithms that provide formal guarantees 
based on strong theoretical foundations. Algorithms we aim for ensure that the
solution is found whenever it exists in a class of functions of interest.
Furthermore, the algorithms ensure that the 
generated program is indeed the program the user wants by
detecting once the solution is unique and otherwise identifying a differentiating
example whose output reduces the space of possible solutions. 

In this paper, we focus on synthesizing printing functions for objects or
algebraic data types (ADT), which are at the core of many programming 
languages.
Converting such structured values to strings is very common, including uses such as pretty printing,
debugging, and serialization.
Writing methods to convert objects to strings is repetitive and usually requires the user to code himself mutually 
recursive \texttt{toString} functions. Although 
some languages have default printing functions, these functions are often not 
adequate.
For example, the object Person(``Joe'', 31) might have to be printed
``Joe is 31 years old'' for better readability, or 
``<td>Joe</td><td>31</td>'' if printed as part of an HTML table.
How \textit{feasible} is it for the computer to learn these ``printing'' 
functions from examples?

The state of the art in this context~\cite{laurence_learning_2014,laurence_phd_2014}
requires the user to provide 
\emph{enough} examples. If the user gives \emph{too few} examples, the 
synthesis algorithm is not guaranteed to return a valid printing function, and 
there is no simple way for the user to know which examples should be added so 
that the synthesis algorithm finishes properly.

Our contribution is to provide an algorithm that is able to determine exactly which
questions to ask the user so that the desired function can be derived. Moreover,
in order to learn a function, our algorithm (Algorithm~\ref{algo3}) only 
needs to ask a linear number of questions 
(as a function of the size of the ADT declaration).

Our results hold for recursive functions that take ADT as input, and 
output strings.
We model these functions by tree-to-string transducers, called 
\emph{single-state sequential top-down tree-to-string transducers}~\cite{boiret_normal_2016,graehl_training_2004,helmut_seidl_equivalence_2015,laurence_learning_2014,staworko_printers_2009}, or \stw{} for
short. In this formalism, objects are represented as labelled trees, and a 
transducer goes through the tree top down in order to display it as a string.
\emph{Single-state} means the transducer keeps no memory as it traverses the 
tree. \emph{Sequential} is a shorthand for \emph{linear} and 
\emph{order-preserving}, meaning that each subtree is printed only once 
(linear), and the subtrees of a node are displayed in order (order-preserving).
In particular, such transducers cannot directly represent recursive functions 
that have extra parameters alongside the tree to print.
Our work on \stw{s} establishes a foundation that may be used for larger 
classes of transducers.

Our goal is to learn a \stw{} from a set of positive input/output examples,
called a \emph{sample}.
We prove the problem of checking whether there exists a \stw{} consistent with 
a given sample is $\NP$-complete in general.
Yet, we prove that when the given sample is closed under subtree, i.e.,
every tree in the sample has all of its subtrees in the sample,
the problem of finding a compatible \stw{} can be solved in polynomial time.
For this, we reduce the problem of checking whether there exists an 
\stw{} consistent with a sample to the problem of solving word equations.
The best known algorithm to solve word equations takes linear space,
and exponential time~\cite{plandowski1999satisfiability,jez2017word}.
However, we prove that the word equations we build are 
of a particular form, which we call \regular{}, and our first algorithm learns 
\stw{s} by solving \regular{} equations in polynomial time.

We then tackle the problem of ambiguities that come from underspecified 
samples. More precisely, it is possible that, given a sample, there exist two 
\stw{s} that are consistent with the sample, but that are not equivalent on a 
domain $\domain$ of trees.
We thus define the notion of \emph{tree test set} of a domain $\domain$, which
guarantees that, any two \stw{s} which are equivalent on the tree test set
are also equivalent on the whole domain $\domain$.
We give a method to build tree test sets of size $O(|\domain|^3)$ from 
a domain of trees given as a non-deterministic top-down automaton.
Our second learning algorithm takes as input a domain $\domain$, 
builds the tree test set of $\domain$, and asks for the user the output 
to all trees in the tree test set.
Our second algorithm then invokes our first algorithm on the given sample.

This construction relies on fundamental results on a known relation between 
sequential top-down tree-to-string transducers and morphisms
(a morphism is a function that maps the concatenation of two words
to the concatenation of their images), and on the
notion of \emph{test set}~\cite{staworko_printers_2009}. 
Informally, a test set of a language of words $L$ is a subset 
$T \subseteq L$ such that any two morphisms which are equivalent on $T$
are also equivalent on $L$.
In the context of \stw{s}, the language $L$ is a context-free language,
intuitively representing the yield of the domain $\domain$ mentioned above.
Prior to our work announced in \cite{DBLP:journals/corr/MayerH16}, the best known construction for a test set of a 
context-free grammar $G$ produced test sets of size $O(|G|^6)$, while
the best known lower bound was 
$O(|G|^3)$~\cite{plandowski_testset_1994,plandowski_testset_1995}.
We show the $O(|G|^3)$ is in fact tight, and give a construction that, given 
any grammar $G$, produces a test set for $G$ of size $O(|G|^3)$.

Finally, our third and, from a practical point of view, the
main algorithm, improves the second one by analyzing the 
previous outputs entered by the user, in order to infer the next output. More 
specifically, the outputs previously entered by the user give constraints
on the transducer being learned, and therefore restrict the possible 
outputs for the next questions. Our algorithm computes these possible outputs
and, when there is only one, skips the question. Our algorithm only asks
the user a question when there are at least two possible outputs for 
a particular input.
The crucial part of this algorithm is to prove that such ambiguities happen 
at most $O(|\domain|)$ times. Therefore, our third algorithm asks the user only 
$O(|\domain|)$ questions, greatly improving our second one that asks 
$O(|\domain|^3)$ questions.
Our result relies on carefully inspecting the word equations produced by 
the input/output examples.

We implemented our algorithms in an open-source tool available 
at \url{https://github.com/epfl-lara/prosy}. In sections~\ref{sec:treewithvalues}
and \ref{sec:implementation}, we describe how to extend our algorithms and 
tool to ADTs which contain String (or Int) as a primitive 
type.
We call the implementation of our algorithms \textit{proactive} synthesis, because it produces a \textit{complete} set of questions ahead-of-time whose answers will help to synthesize a unique tree-to-string function,
\textit{filters out} future questions whose answer could be actively inferred after each user's answer,
and produces \textit{suggestions} as multiple choice or pre-filled answers to minimize the answering effort.

\subsection*{Contributions}

Our paper makes the following contributions:
\begin{enumerate}
\item
    A new efficient algorithm to synthesize recursive functions from 
    examples. We give a polynomial-time algorithm to obtain a \stw{}
    from a sample \emph{closed under subtree}. 
    When the sample is not necessarily closed under subtree, we prove that 
    the problem of checking whether there exists a \stw{} consistent 
    with the sample is $\NP$-complete
    (Section~\ref{sec:learningstw}). This result is based on a 
    fundamental contribution:
\begin{itemize}
\item 
\label{poly}
    A polynomial-time algorithm for solving a class of word equations that come from
    a synthesis problem (\textit{\regular{}} word equations, 
    Section~\ref{sec:learningstw}).
\end{itemize}
\item
    An algorithm that synthesize recursive functions without ambiguity by 
    generating an exhaustive set of questions to ask to the user,
    in the sense that any two recursive functions that agree on these inputs, 
    are equal on their entire domain
    (Section~\ref{sec:learningwithoutambiguity}).
    This is based on the following fundamental contribution:
\begin{itemize}
\item
\label{testset}
    A constructive upper bound of $O(|G|^3)$ on the size of a test set for a
    context-free grammar $G$, improving on the previous known bound of
    $O(|G|^6)$~\cite{plandowski_testset_1994,plandowski_testset_1995}
    (Section~\ref{sec:learningwithoutambiguity}).
\end{itemize}
\item
    A proactive and efficient algorithm that synthesizes recursive functions, 
    which 
    only requires the user to enter outputs for the inputs determined by the 
    algorithm.
    Formally, we present
    an interactive algorithm to learn a \stw{} for a domain of
    trees, with the guarantee that the obtained \stw{} is functionally 
    unique. Our algorithm asks the user only a \emph{linear} number of questions 
    (Section~\ref{sec:learningstwinteractively}).
\item
    A construction of a linear tree test set for data types with Strings, which enables 
    constructing a small set of inputs that distinguish between two recursive 
    functions (Section~\ref{sec:treewithvalues}).
\item
  An implementation of our algorithms as an interactive command-line tool
  (Section~\ref{sec:implementation})
\end{enumerate}
We note that the fundamental contributions of (\ref{poly}) and (\ref{testset}) 
are new general results about formal languages and may be of interest on their 
own.

For readability purposes, we only show proof sketches and intuition; detailed proofs are located in the Appendices.
\section{Example Run of Our Synthesis Algorithm}
\label{sec:walkthrough}
\label{section:walkthrough}

To motivate our problem domain, we present a run of our algorithm on 
an example. The example is an ADT representing a context-free grammar. 
It defines its custom alphabet (\textsf{Char}), words (\textsf{CharList}), and 
non-terminals indexed by words (\textsf{NonTerminal}).
A rule (\textsf{Rule}) is a pair made of a non-terminal and a sequence of 
symbols (\textsf{ListSymbol}), which can be non-terminals or terminals 
(\textsf{Terminal}).
Finally, a grammar is a pair made of a (starting) non-terminal and a sequence 
of rules.

The input of our algorithm is the following file (written in Scala syntax):
\begin{lstlisting}
abstract class Char
case class a() extends Char
case class b() extends Char

abstract class CharList
case class NilChar() extends CharList
case class ConsChar(c: Char, l: CharList) extends CharList

abstract class Symbol
case class Terminal(t: Char) extends Symbol
case class NonTerminal(s: CharList) extends Symbol

case class Rule(lhs: NonTerminal, rhs: ListSymbol)

abstract class ListRule
case class ConsRule(r: Rule, tail: ListRule) extends ListRule
case class NilRule() extends ListRule

abstract class ListSymbol
case class ConsSymbol(s: Symbol, tail: ListSymbol) extends ListSymbol
case class NilSymbol() extends ListSymbol

case class Grammar(s: NonTerminal, r: ListRule)
\end{lstlisting}

We would like to synthesize a recursive tree-to-string function \textsf{print}, such that if we compute, for example:
\begin{lstlisting}
print(Grammar(NonTerminal(NilChar()),
    ConsRule(Rule(NonTerminal(NilChar()),
      ConsSymbol(Terminal(a()),
        ConsSymbol(NonTerminal(NilChar()),
          ConsSymbol(Terminal(b()), NilSymbol())))),
      ConsRule(Rule(NonTerminal(NilChar()),
          NilSymbol())), NilRule())))
\end{lstlisting}
the result should be:
\begin{lstlisting}
Start: N
N -> a N b
N ->
\end{lstlisting}
We would like the \textsf{print} function to handle any valid \textsf{Grammar} tree.

When given these class definitions above, our algorithm precomputes a set of terms from the ADT,
so that any two single-state recursive functions which output the same Strings
for these terms also output the same Strings for any term from this ADT.
(This is related to the notion of \emph{tree test set} defined in 
Section~\ref{subsection:treetestset}.) 
Our algorithm will determine the outputs for these terms by interacting 
with the user and asking questions.
Overall, for this example, our algorithm asks the output 
for 14 terms.%

\newcommand{\enterkey}{\ensuremath{\hookleftarrow}}
\newcommand{\theuserenters}[1]{%
\userinput{#1}\enterkey{}\quad
}
\newcommand{\theuserentersraw}[1]{#1\enterkey{}\quad}
For readability, question lines provided by the synthesizer are indented. Lines entered by the user finish by the symbol $\enterkey{}$, meaning that she pressed the ENTER key. Everything after $\enterkey{}$ on the same line is our comment on the interaction. ``It'' usually refers to the synthesizer. After few interactions, the questions themselves are shortened for conciseness.
The interaction is the following:
\begin{lstlisting}
  Proactive Synthesis.
  If you ever want to enter a $\text{new}$ line, terminate your line by \ and press Enter.
  What should be the function output $\text{for}$ the following input tree?
  a
a$\enterkey{}$
  What should be the function output $\text{for}$ the following input tree?
  b
b$\enterkey{}$
  NilChar ?
$\enterkey{}$ $\hspace{2cm}\textit{indeed, \userinput{NilChar} is an empty string.}$
  NilSymbol ?
$\enterkey{}$ $\hspace{2cm}\textit{No symbol at the right-hand-side of a rule}$
  NilRule ?
$\enterkey{}$ $\hspace{2cm}\textit{No rule left describing the grammar}$
  What should be the function output $\text{for}$ the following input tree?
  Terminal(a)
  Something of the form: [...]a[...]
a$\enterkey{}$ $\hspace{1.9cm}\textit{Terminals contain only one char. Note the hint provided by the synthesizer.}$
  NonTerminal(NilChar) ?
N$\enterkey{}$
  ConsChar(b,NilChar) ? Something of the form: [...]b[...]
b$\enterkey{}$ $\hspace{1.9cm}\textit{A \textsf{ConsChar} is a concatenation of a char and a string}$
  What should be the function output $\text{for}$ the following input tree?
  NonTerminal(ConsChar(b,NilChar))
  1) Nb
  2) bN
  Please enter a number between 1 and 2, or 0 $\text{if}$ you really want to enter your answer manually
1$\enterkey{}$ $\hspace{1.9cm}\textit{Note that it was able to infer only two possibilities, thus the closed question.}$
  Grammar(NonTerminal(NilChar),NilRule) ? Something of the form: [...]N[...]
Start: N$\enterkey{}$
  ConsSymbol(Terminal(a),NilSymbol) ? Something of the form: [...]`a`[...]
 a$\enterkey{}$ $\hspace{1.8cm}\textit{Symbols on the right-hand-side of a \textsf{Rule} are prefixed with a space}$
Rule(NonTerminal(NilChar),NilSymbol) ? Something of the form: [...]N[...]
N ->$\enterkey{}$ $\hspace{1.25cm}\textit{A rule with no symbols on the right-hand-side}$
  ConsRule(Rule(NonTerminal(NilChar),NilSymbol),NilRule) ?
  Something of the form: [...]N ->[...]
\$\enterkey{}$ $\hspace{1.9cm}\textit{A newline}$
N ->$\enterkey{}$
  What should be the function output $\text{for}$ the following input tree?
  Rule(NonTerminal(NilChar),ConsSymbol(Terminal(`a`),NilSymbol))
  1) N  `a`->
  2) N - `a`>
  3) N -> `a`
  4) N `a` ->
  Please enter a number between 1 and 4, or 0 $\text{if}$ you really want to enter your answer manually
3$\enterkey{}$
\end{lstlisting}
The synthesizer %
then emits the desired recursive tree-to-string function, along with a complete set of the tests that determine it:
\begin{lstlisting}
def print(t: Any): String = t match {
  case a() $\Rightarrow$ "a"
  case b() $\Rightarrow$ "b"
  case NilChar() $\Rightarrow$ ""
  case ConsChar(t1,t2) $\Rightarrow$ print(t1) + print(t2)
  case Terminal(t1) $\Rightarrow$ "`" + print(t1) + "`"
  case NonTerminal(t1) $\Rightarrow$ "N" + print(t1)
  case Rule(t1,t2) $\Rightarrow$ print(t1) + " ->" + print(t2)
  case ConsRule(t1,t2) $\Rightarrow$ "\n" + print(t1) + print(t2)
  case NilRule() $\Rightarrow$ ""
  case ConsSymbol(t1,t2) $\Rightarrow$ " " + print(t1) + print(t2)
  case NilSymbol() $\Rightarrow$ ""
  case Grammar(t1,t2) $\Rightarrow$ "Start: " + print(t1) + print(t2)
} // the part below is a contract, not needed to execute the recursive function
ensuring { (res: string) => res == (t match {
  case a() => "a"
  case b() => "b"
  case NilChar() => ""
  case NilSymbol() => ""
  case NilRule() => ""
  case Terminal(a()) => "a"
  case NonTerminal(NilChar()) => "N"
  case ConsChar(b(),NilChar()) => "b"
  case NonTerminal(ConsChar(b(),NilChar())) => "Nb"
  case Grammar(NonTerminal(NilChar()),NilRule()) => "Start: N"
  case ConsSymbol(Terminal(a()),NilSymbol()) => " a"
  case Rule(NonTerminal(NilChar()),NilSymbol()) => "N ->"
  case ConsRule(Rule(NonTerminal(NilChar()),NilSymbol()),NilRule()) => "\nN ->"
  case Rule(NonTerminal(NilChar()),ConsSymbol(Terminal(a()),NilSymbol())) => "N -> a"
  case _ => res})
}
\end{lstlisting}

Observe that, in addition to the program, the synthesis system emits as a postcondition (after the \texttt{ensuring} construct) the
recorded input/output examples (tests).
Our work enables the construction of an IDE that would automatically maintain the bidirectional correspondence between the body of the recursive
function and the postcondition
that specifies its input/output tests.
If the user modifies an example in the postcondition, the system could re-synthesize the function, asking for clarification in cases
where the tests become ambiguous.
If the user modifies the program, such system can regenerate the tests.

Depending on user's answers, the total number of questions that the synthesizers asks varies (see section~\ref{sec:evaluation}). Nonetheless, the properties that we proved for our algorithm guarantee that the number of questions remains at most \emph{linear} as a function of the size of the algebraic data type declaration.

When the user enters outputs which are not consistent, i.e., for which 
there exists no printing function in the class of functions that we consider, our tool directly detects it and warns 
the user. For instance, for the tree \userinput{ConsRule(Rule(NonTerminal(NilChar),NilSymbol),NilRule)},
if the user enters \userinput{N- >} with the space and the dash inverted, the system detects that this output 
is not consistent with the output provided for tree \userinput{Rule(NonTerminal(NilChar),NilSymbol)}, and 
asks the question again.

\begin{lstlisting}
We cannot have the transducer convert ConsRule(Rule(NonTerminal(NilChar),NilSymbol),NilRule)
to N- >.
Please enter something consistent with what you previously entered (e.g. 'N ->','N ->bar',...)?
\end{lstlisting}

\section{Discussion}

\subsection{Advantages of Synthesis Approach}

It is important to emphasize that in the approach we outline, the
developer not only enters less text in terms of the number of characters
than in the above source code, but that the input from the user is
entirely in terms of concrete input-output \emph{values},
which can be easier to reason about for non-expert users
than recursive programs with variable names and
control-flow.

It is notable that the synthesizer in
many cases offered suggestions, which means that the user
often simply needed to check whether one of the candidate
outputs is acceptable. Even in cases where the user needed
to provide new parts of the string, the synthesizer in many
cases guided the user towards a form of the output
consistent with the outputs provided so far. Because of this
knowledge, the synthesizer could also be stopped early by,
for example, guessing the unknown information according to
some preference (e.g. replacing all unknown string constants
by empty strings), so the user can in many cases obtain a
program by providing a very small amount of information.

Such easy-to-use interactions could be implemented as a
pretty printing wizard in an IDE, for example triggered when
the user starts to write a function to convert an ADT to a
String.

Our experience in writing pretty printers manually suggests
that they often require testing to ensure that the generated
output corresponds to the desired intuition of the
developer, suggesting that input-output tests may be a
better form of specification even if in cases where they are
more verbose. We therefore believe that it is valuable
to make available to users and developers sucn an alternative
method of specifying recursive functions, a method that
can co-exist with the conventional explicitly written recursive
functions and the functions derived automatically (but generically) by the compiler
(such as default printing of algebraic data type values in Scala),
or using polytypic programming approaches \cite{JanssonPolytypic} and
serialization libraries \cite{DBLP:conf/oopsla/MillerHBO13}.
(Note that the generic approaches can reduce the boilerplate,
but do not address the problem of unambiguously generalizing \emph{examples} to
recursive functions.) 

\subsection{Challenges in Obtaining Efficient Algorithms}

The problem of inferring a program from examples requires
recovering the constants embedded in the program from the
results of concatenating these constants according to the
structure of the given input tree examples. This presents
two main challenges. The first one is that the algorithm
needs to split the output string and identify which parts
correspond to constants and which to recursive calls. This
process becomes particularly ambiguous if the alphabet used is small
or if some constants are empty strings. A natural way to solve
such problems is to formulate them as a conjunction of word equations.
Unfortunately, the best known deterministic algorithms for solving word equations
run in exponential time (the best complexity upper bound for the problem
takes linear space~\cite{plandowski1999satisfiability,jez2017word}). Our paper shows that, under an assumption
that, when specifying printing of a tree, we also specify printing of its subtrees,
we obtain word equations solvable in \emph{polynomial time}.

The next challenge is the number of examples that need to be
solved. Here, a previous upper bound derived from the theory
of test sets of context-free languages was $\Omega(n^6)$,
which, even if polynomial, results in impractical number of
user interactions.  In this paper we improve this
theoretical result and show that tests sets are in fact in
$O(n^3)$, asymptotically matching the known lower bound.

Furthermore, if we allow the learning algorithm to choose
the inputs one by one after obtaining outputs, the overall
learning algorithm has a \emph{linear} number of queries to
user and to equation solving subroutine, as a function of
the size of tree data type definition. Our contributions
therefore lead to tools that have completeness guarantees
with much less user input and a shorter running time than
the algorithms based on prior techniques.

We next present our algorithms as well as the results
that justify their correctness and completeness.
\section{Notation}

We start by introducing our notation and terminology for some standard concepts.
Given a (partial) function from 
$f: A \to B$, and a set $C$, 
$\restrict{f}{C}$ denotes the (partial) function 
$g: A \cap C \to B$ such that $g(a) = f(a)$ for all
$a \in A \cap C$.

A word (string) is a finite sequence of elements of a finite set $\Sigma$,
which we call an \emph{alphabet}. 

A \emph{morphism} $f: \Sigma^* \rightarrow \Gamma^*$ is a function
such that $f(\varepsilon) = \varepsilon$ and for every
$u,v \in \Sigma^*$, $f(u \cdot v) = f(u) \cdot f(v)$,
where the symbol `$\cdot$' denotes the concatenation of words (strings).

A \emph{non-deterministic finite automaton (NFA)} is a tuple 
$(\Output,Q,q_i,F,\delta)$ where 
$\Output$ is the alphabet, $Q$ is the set of states, $q_i \in Q$ is the 
initial state, $F$ is the set of final states,
$\delta \subseteq Q \times \Output \times Q$ is the transition relation.
When the transition relation is deterministic, that is 
for all $q,p_1,p_2 \in Q, a \in \Output$, 
if $(q,a,p_1) \in \delta$ and $(q,a,p_2) \in \delta$, then $p_1 = p_2$,
we say that $A$ is a \emph{deterministic finite automaton (DFA)}.

A \emph{context-free grammar}
  $\gram$ is a tuple $(\NTerm,\Term,\Prod,\Start)$ where:
\begin{itemize}
\item $\NTerm$ is a set of \emph{non-terminals},
\item $\Term$ is a set of \emph{terminals}, disjoint from $\NTerm$,
\item $\Prod \subseteq \NTerm \times (\NTerm \cup \Term)^* $ 
  is a set of \emph{production rules},
\item $\Start \in \NTerm$ is the starting non-terminal symbol.
\end{itemize}
A production $(\nterm,\rhs) \in \Prod$ is denoted $\mkrule{\nterm}{\rhs}$.
The \emph{size} of $\gram$, denoted $|\gram|$, is the sum of sizes of
each production in $\Prod$:
    $\sum_{\mkrule{\nterm}{\rhs} \in \Prod} 1+|\rhs|$.
A grammar is \emph{linear} if for every production 
$\mkrule{\nterm}{\rhs} \in \Prod$, the $\rhs$ string contains at most one 
occurrence of $\NTerm$.
By an abuse of notation, we denote by $\gram$ the set of words
produced by $\gram$.

\subsection{Trees and Domains}

A \emph{ranked alphabet} $\Sigma$ is a set of pairs 
$(f,k)$ where $f$ is a symbol from a finite alphabet, and $k \in \Nat$.
A pair $(f,k)$ of a ranked alphabet is also denoted $\ranked{f}{k}$.
We say that symbol $f$ has a \emph{rank} (or \emph{arity}) equal to $k$.
We define by $\TreeA{\Sigma}$ the set of trees defined over alphabet $\Sigma$. 
Formally,  $\TreeA{\Sigma}$ is the smallest set such that, if 
$t_1,\dots,t_k \in \TreeA{\Sigma}$, and
$\ranked{f}{k} \in \Sigma$ for some $k \in \Nat$, 
then $f(t_1,\dots,t_k) \in \TreeA{\Sigma}$.
A set of trees $T$ is \emph{closed under subtree} if 
for all $f(t_1,\dots,t_k) \in T$, 
for all $i \in \set{1,\dots,k}$,
$t_i \in T$.

A top-down tree automaton $T$ is a tuple $(\Sigma,\State,I,\delta)$ where 
$\Sigma$ is a ranked alphabet, $I \subseteq \State$ is the set of initial 
states, and $\delta \subseteq \Sigma \times \State \times \State^*$.
The set of trees $\getlang{T}$ recognized by $T$ is defined recursively
as follows.
For $\ranked{f}{k} \in \Sigma$, 
$q \in \State$, and 
$t = f(t_1,\dots,t_k) \in \TreeA{\Sigma}$, 
we have $t \in \getlang{T}_q$
iff there exists 
$(f,q,q_1 \cdots q_k) \in \delta$ such that for $1 \leq i \leq k$,  
$t_i \in \getlang{T}_{q_i}$. The set 
 $\getlang{T}$ is then defined as $\bigcup_{q \in I} \getlang{T}_q$.

Algebraic data types are described by the notion of \emph{domain}, 
which is a set of trees recognized by a top-down tree automaton
$T = (\Sigma,\State,I,\delta)$.
The \emph{size} of the domain is the sum of sizes of each transition in 
$\delta$, that is
$\sum_{(\ranked{f}{k},q,q_1 \cdots q_k) \in \delta} 1 + k$.

\begin{example}\label{example:htmldomain}
  In this example and the following ones, we illustrate our notions using an encoding of HTML-like data structures.
  Consider the following algebraic data type definitions in Scala:
\begin{lstlisting}
abstract class $\nodeType$
case class $\node$(t: $\tagType$, l: $\listType$) extends $\nodeType$

abstract class $\tagType$
case class $\ddiv$() extends $\tagType$
case class $\dpre$() extends $\tagType$
case class $\dspan$() extends $\tagType$

abstract class $\listType$
case class $\cons$(n: $\nodeType$, l: $\listType$) extends $\listType$
case class $\nnil$() extends $\listType$
\end{lstlisting}

The corresponding domain \exampledomain{} is described by the following:
\begin{align*}
\Sigma = \{ & \nnil^{(0)}, \cons^{(2)}, \node^{(2)}, \ddiv^{(0)}, \dpre^{(0)}, \dspan^{(0)}\} \\
Q = \{& \nodeType, \tagType, \listType\} \\
I = \{& \nodeType, \tagType, \listType\} \\
\delta  = \{ & (\node, \nodeType, (\tagType, \listType)),\\
&(\ddiv, \tagType, ()), (\dpre, \tagType, ()), (\dspan, \tagType, ()), \\
&(\cons, \listType, (\nodeType, \listType)),\\
&(\nnil, \listType, ())
\}
\end{align*}

\end{example}

\subsection{Transducers}
\newcommand{\AssociateConstants}{\ensuremath{\delta}}
A \emph{deterministic, sequential, single-state,
top-down tree-to-string transducer} $\trans$
  (\stw{} for short)
is a tuple $(\Alphabet, \Output, \AssociateConstants)$ where:
\begin{itemize}
\item $\Sigma$ is a ranked alphabet (of trees),
\item $\Output$ is an alphabet (of words),
\item 
    $\AssociateConstants$ is a function over $\Sigma$ such that 
    $\forall \ranked{f}{k} \in \Sigma.\ 
    \AssociateConstants(f) \in (\Output^*)^{k+1}$.
\end{itemize}

Note that the transducer does not depend on a particular domain for $\Sigma$,
but instead can map any tree from $\TreeA{\Sigma}$ to a word.
Later, when we present our learning algorithms for \stw{s}, we restrict 
ourselves to particular domains provided by the user of the algorithm.

We denote by $\semantics{\trans}$ the function from trees to words
associated with the \stw{} $\trans$. 
Formally, for every $\ranked{f}{k} \in \Sigma$, we have
$\semantics{\trans}(f(t_1,\dots,t_k)) = 
  u_0 \cdot \semantics{\trans}(t_1) \cdot u_1 \cdots \semantics{\trans}(t_k) \cdot u_k$
if $\AssociateConstants(f) = (u_0,  u_1, \dots, u_k)$.
When clear from context, we abuse notation and use $\trans$ as a shorthand
for the function $\semantics{\trans}$.

\begin{example}\label{example:transducer}
A transducer $\trans=(\Alphabet, \Output, \AssociateConstants)$ converting HTML trees into a convenient syntax for some programmatic templating engines\footnote{\url{https://github.com/lihaoyi/scalatags}} may be described by: 
\begin{align*}
\Sigma = & \omit\rlap{$\set{
    \nnil^{(0)}, \cons^{(2)}, \node^{(2)},
    \ddiv^{(0)}, \dpre^{(0)}, \dspan^{(0)}\}}$} \\
\Output = & [\textit{All symbols}] \\
\AssociateConstants(\node) =& (\text{``<.''}, \varepsilon, \varepsilon) \\
\AssociateConstants(\ddiv) =& (\text{``div''}) &
\AssociateConstants(\dpre) =& (\text{``pre''}) &
\AssociateConstants(\dspan) &= (\text{``span''}), \\
\AssociateConstants(\cons) =& (\text{``(''}, \text{``)''}, \varepsilon) &
\AssociateConstants(\nnil) =& (\varepsilon)
\end{align*}
In Scala, this is written as follows:
\begin{lstlisting}
def tau(input: Tree) = input match {
  case $\node$(t, l)  $\Rightarrow$ "<." + tau(t) + "" + tau(l) + ""
  case $\ddiv$() $\Rightarrow$ "div"
  case $\dpre$() $\Rightarrow$ "pre"
  case $\dspan$() $\Rightarrow$ "span"
  case $\cons$(n, l) $\Rightarrow$ "(" + tau(n) + ")" + tau(l) + ""
  case $\nnil$() $\Rightarrow$ ""
}
\end{lstlisting}
For example, \lstinline|tau(node(div,cons(node(span,nil,cons(node(pre,nil)))))) = "<.div(<.span())(<.pre())"|
\end{example}
\section{Transducers as Morphisms}
\addtocounter{subsection}{1}

For a given alphabet $\Sigma$, a \stw{} $(\Alphabet, \Output, \AssociateConstants)$
is completely determined by the constants that appear in $\AssociateConstants$.
This allows us to define a one-to-one correspondence between transducers and 
morphisms.
This correspondence is made through what we call the \emph{default transducer}.
More specifically,
$\Output$ is the set 
$\SigmaPairs = \set{(f,i)\ |\ \ranked{f}{k} \in \Sigma \land 0 \leq i \leq k}$
and for all $\ranked{f}{k} \in \Sigma$,
we have $\AssociateConstants(f) = ( (f,0), (f,1), \dots, (f,k) )$.
The default transducer produces sequences of pairs from $\SigmaPairs$.

\begin{example}\label{example:defaulttransducer} For 
$\Sigma = \{ \nnil^{(0)}, \cons^{(2)}, \node^{(2)}, \ddiv^{(0)}, \dpre^{(0)}, \dspan^{(0)}\}$, $\default{\Sigma}$ is:

\begin{align*}
\Output = & \omit\rlap{\{ (\node, 0), (\node, 1), (\node, 2), (\ddiv, 0), (\dpre, 0), (\dspan, 0)} \\
&\omit\rlap{(\cons, 0), (\cons, 1), (\cons, 2), (\nnil, 0)\}} \\
\AssociateConstants(\node) =\ & ((\node, 0), (\node, 1), (\node, 2)) \\
\AssociateConstants(\ddiv) =\ & (\ddiv, 0) &
\AssociateConstants(\dpre) =\ & (\dpre, 0) &
\AssociateConstants(\dspan) =\ & (\dspan, 0) \\
\AssociateConstants(\cons) =\ & ((\cons, 0), (\cons, 1), (\cons, 2)) &
\AssociateConstants(\nnil) =\ & (\nnil, 0))
\end{align*}

In Scala, $\default{\Sigma}$ can be written as follows 
($+$ is used to concatenate elements and lists):
\begin{lstlisting}
def tauSigma(input: Tree): List[$\SigmaPairs$] = input match {
  case $\node$(t, l) $\Rightarrow$ ($\node$,0) + tauSigma(t) + ($\node$,1) + tauSigma(l) + ($\node$,2)
  case $\ddiv$() $\Rightarrow$ ($\ddiv$,0)
  case $\dpre$() $\Rightarrow$ ($\dpre$,0)
  case $\dspan$() $\Rightarrow$ ($\dspan$,0)
  case $\cons$(n, l) $\Rightarrow$ ($\cons$,0) + tauSigma(n) + ($\cons$,1) + tauSigma(l) + ($\cons$,2)
  case $\nnil$(n, l) $\Rightarrow$ ($\nnil$,0)
}
\end{lstlisting}
\end{example}

\begin{restatable}{lem}{injective}
\label{lemma:injective}
For any ranked alphabet $\Sigma$, the function
$\semantics{\default{\Sigma}}$ is injective.
\end{restatable}

\begin{figure} 

\begin{lstlisting}
def tree(w: List[$\SigmaPairs$]): Tree =
  if w is empty or does not start with some (f, 0):
    throw error
  let (f, 0) = w.head
  w $\leftarrow$ w.tail
  for i from 1 to arity(f)
    $t_i$ = tree(w)
    assert(w starts with (f, i))
    w $\leftarrow$ w.tail
  return f($t_1$, $\ldots$, $t_k$)
\end{lstlisting}

\caption{Parsing algorithm to obtain $\gettree(w)$ from a word 
$w \in \SigmaPairs^*$. When the algorithm fails, because of a pattern matching
error or because of the thrown exception, 
it means there exists no $t$ such that $\default{\Sigma}(t) = w$.
\label{figure:tree}
}

\end{figure}
 
Following Lemma~\ref{lemma:injective}, 
for a word $w \in \SigmaPairs^*$, 
we define $\gettree(w)$ to be the unique tree (when it exists)
such that $\default{\Sigma}(\gettree(w)) = w$.
We show in Figure~\ref{figure:tree} how to obtain $\gettree(w)$ in 
linear time from $w$.

For a \stw{} $\trans = (\Alphabet, \Output, \AssociateConstants)$, we define the morphism
$\getmorph{\trans}$ from $\SigmaPairs$ to $\Output^*$, and such that, 
for all $\ranked{f}{k} \in \Sigma$,
  $i \in \set{0,\dots,k}$, 
  $\getmorph{\trans}(f,i) = u_i$
  where $\AssociateConstants(f) = (u_0, u_1, \dots, u_k)$.
Conversely, 
given a morphism $\morph: \SigmaPairs \to \Output^*$,
we define $\gettransducer{\morph}$ as
$\default{\Sigma}$ where each output $l \in \SigmaPairs$
is replaced by $\morph(l)$.

\begin{example}\label{example:morph}
For Example~\ref{example:transducer}, $\getmorph{\trans}$ is defined by:
\begin{align*}
\getmorph{\trans}(\node, 0) &= \text{``<.''} & \getmorph{\trans}(\cons, 0) &= \text{``(''} \\
\getmorph{\trans}(\node, 1) &= \varepsilon & \getmorph{\trans}(\cons, 1) &= \text{``)''} \\
\getmorph{\trans}(\node, 2) &= \varepsilon & \getmorph{\trans}(\cons, 2) &= \varepsilon \\
\getmorph{\trans}(\ddiv, 0) &= \text{``div''}  & \getmorph{\trans}(\nnil, 0) &= \varepsilon \\
\getmorph{\trans}(\dpre, 0) &= \text{``pre''} &
\getmorph{\trans}(\dspan, 0) &= \text{``span''}
\end{align*}
\end{example}

Note that for any morphism: $\morph : \SigmaPairs \to \Output^*$,
$\getmorph{\gettransducer{\morph}} = \morph$ and
for any \stw{} $\trans$, 
$\gettransducer{\getmorph{\trans}} = \trans$. Moreover, we have the following 
result, which expresses the output of a \stw{} $\trans$ using the morphism 
$\getmorph{\trans}$.

\begin{lem}
\label{lemma:default}
For a \stw{} $\trans$, 
and for all $t \in \TreeA{\Sigma}$, 
$\getmorph{\trans}(\default{\Sigma}(t)) = \trans(t)$.
\end{lem}
\begin{proof}
Follows directly from the definitions of 
$\getmorph{\trans}$ and $\default{\Sigma}$.
\end{proof}
\begin{example}
Let $t= \cons(\node(\ddiv,\nnil),\nnil)$.
For $\getmorph{}$ defined as in Example~\ref{example:morph} and the transducer 
$\trans$ as in Example~\ref{example:transducer}, the left-hand-side of the 
equation of Lemma~\ref{lemma:default} translates to:

\begin{align*}
&\getmorph{\trans}(\default{\Sigma}(t)) \\
=\ &\getmorph{\trans}(\default{\Sigma}(\cons(\node(\ddiv,\nnil),\nnil))) \\
=\ &\getmorph{\trans}((\cons, 0)(\node, 0)(\ddiv,0)(\node,1)(\nnil,0)(\node,2)(\cons,1)(\nnil,0)(\cons,2)) \\
=\ & \text{``(''}\cdot\text{``<.''}\cdot\text{``div''}\cdot\varepsilon\cdot\varepsilon\cdot\varepsilon\cdot\text{``)''}\cdot\varepsilon\cdot\varepsilon\\
=\ & \text{``(<.div)''}
\end{align*}

Similarly, the right-hand-side of the equation can be computed as follows:
\begin{align*}
 & \trans(t) \\
=\ & \trans(\cons(\node(\ddiv,\nnil),\nnil))\\
=\ & \text{``(''}\cdot\trans(\node(\ddiv,\nnil))\cdot\text{``)''}\cdot\trans(\nnil)\cdot\varepsilon \\
=\ & \text{``(''}\cdot\text{``<.''}\cdot\trans(\ddiv)\cdot\varepsilon\cdot\trans(\nnil)\cdot\varepsilon\cdot\text{``)''}\cdot\varepsilon\cdot\varepsilon \\
=\ & \text{``(<.div)''}
\end{align*}

\end{example}

We thus obtain that checking equivalence of \stw{s} can be reduced to checking 
equivalence of morphisms on a context-free language.

\begin{lem}[See \cite{staworko_printers_2009}]
\label{lemma:stwtomorphism}
Let $\trans_1$ and $\trans_2$ be two \stw{s}, and 
$\domain = (\Sigma,\State,I,\delta)$ a domain.
Then 
$\restrict{\semantics{\trans_1}}{\domain} = 
\restrict{\semantics{\trans_2}}{\domain}$
if and only if 
$\restrict{\getmorph{\trans_1}}{G} = 
\restrict{\getmorph{\trans_2}}{G}$
where $G$ is the context-free language 
$\set{\default{\Sigma}(t)\ |\ t \in \domain}$.
\end{lem}

\begin{proof}
Follows from Lemma~\ref{lemma:default}. $G$ is context-free, as it
can be recognized by the grammar 
$(\NTerm_\gram,\SigmaPairs,\Prod_\gram,\Start_\gram)$ where:
\begin{itemize}
\item
    $\NTerm_\gram = \set{\Start_\gram}\, \cup\, \set{A_{q}\ |\ q \in \State}$,
    where $\Start_\gram$ is a fresh symbol used as the starting non-terminal,
\item The productions are:\\
  $\begin{array}{ll}
  \Prod_\gram = &
    \{
    A_{q} \rightarrow\
      (f,0) \cdot
      A_{q_1} \cdot
      (f,1) \cdots
      A_{q_k} \cdot
      (f,k)
      \ |\ 
       \ranked{f}{k} \in \Sigma\ \land
        (q, f, (q_1, \dots, q_k)) \in \delta
    \} \\ & \, \cup \,
    \set{ \mkrule{\Start_\gram}{A_q}\ |\ q \in I}
  \end{array}$
\end{itemize}

Note that the size of $G$ is linear in the size of $|\domain|$
(as long as there are no unused states in $\domain$).
\end{proof}

\section{Learning \stw{} from a Sample}\label{sec:learningstw}

We now present a learning algorithm for learning \stw{s} from sets of 
input/output examples, or a \emph{sample}.
Formally, a sample $\sample: \TreeA{\Sigma} \pto \Output^*$
is a partial function from trees to words,
or alternatively, a set of pairs $(t,w)$ with $t \in \TreeA{\Sigma}$ and 
$w \in \Output^*$ such that each $t$ is paired with at most one $w$.

\subsection{NP-completeness of the general case}

In general, we prove that finding whether there exists a
\stw{} consistent with a given a sample is an $\NP$-complete problem.
To prove $\NP$-hardness, we reduce the one-in-three positive SAT
problem. This problem asks, given a formula $\varphi$ with no negated 
variables, whether there exists an assignment such that for each clause 
of $\varphi$, exactly one variable (out of three) evaluates to true.

\begin{restatable}{thm}{npcomplete}
\label{theorem:npcomplete}
Given a sample $\sample$, checking whether there exists a \stw{} $\trans$ such 
that for all $(t,w) \in \sample$, $\trans(t) = w$ is an $\NP$-complete 
problem.
\end{restatable}

\begin{proof}
(Sketch)
We can check for the existence of $\trans$ in $\NP$ using the following idea. 
Every input/output example from the sample gives constraints on the 
constants of $\trans$.
Therefore, to check for the existence of $\trans$, it is sufficient to 
non-deterministically guess constants which are subwords of the given output 
examples.
We can then verify in polynomial-time whether the guessed constants form 
a \stw{} $\trans$ which is consistent with the sample $\sample$.

To prove $\NP$-hardness, we consider a formula $\varphi$, 
instance of the one-in-three positive SAT.
The formula $\varphi$ has no negated variables, and is satisfiable if there 
exists an assignment to the boolean variables such that for each clause 
of $\varphi$, exactly one variable (out of three) evaluates to true.

We construct a sample $\sample$ such that 
there exists a \stw{} $\trans$ 
such that for all $(t,w) \in \sample$, $\trans(t) = w$
if and only if $\varphi$ is satisfiable.
For each clause $(x,y,z) \in \varphi$, we construct 
an input/output example of the form $\sample(x(y(z(\lf)))) = a \#$
where $x$, $y$ and $z$ are symbols of arity $1$ corresponding to the 
variables of the same name in $\varphi$, $\lf$ is a symbol of arity $0$,
and $a$ and $\#$ are two special characters.
Moreover, we add an input/output example stating that 
$\sample(\lf) = \#$.

This construction forces the fact that a \stw{} $\trans$
consistent with $\sample$ will have a non-empty output ($a$) for exactly one 
symbol out of $x$, $y$, and $z$ (therefore matching the requirements of 
one-in-three positive SAT formulas).
\end{proof}
 
In the sequel, we prove that if the domain of the given sample is closed 
under subtree, this problem can be solved in polynomial time.

\subsection{Word Equations}
\newcommand{\x}[1]{\ensuremath{X_{#1}}}
\newcommand{\y}[1]{\ensuremath{y_{#1}}}

Our learning algorithm relies on reducing the problem of learning a \stw{} 
from a sample
to the problem of solving word equations.
In general, the best known algorithm for solving word equations 
is in linear space~\cite{plandowski1999satisfiability,jez2017word},
and takes exponential time to run.
When the domain of the sample $\sample$ is closed under subtree, 
the equations we construct have a particular form, and we call 
them \emph{\regular{} formulas}.
We show there is a polynomial-time algorithm for checking whether 
a \regular{} word formula is satisfiable.

\begin{definition}
Let $\Var$ be a finite set of \emph{variables}, and $\Output$ a finite 
alphabet. A \emph{word equation} $\eq$ is a pair $\concat_1 = \concat_2$
where $\concat_1,\concat_2 \in (\Var \cup \Output)^*$.
A \emph{word formula} $\form$ is a conjunction of word equations.
An \emph{assignment} is a function from $\Var$ to $\Output^*$,
and can be seen as 
a morphism $\assignment: (\Var \cup \Output) \to \Output^*$ 
such that $\assignment(a) = a$ for all $a \in \Output$.

A word formula is satisfiable if there exists an assignment
$\morph: (\Var \cup \Output) \to \Output^*$ 
such that for all equations $\concat_1 = \concat_2$ in $\form$, 
$\morph(\concat_1) = \morph(\concat_2)$.
\end{definition}

A word formula $\form$ is called \emph{\regular{}} if:
1) for each equation $\concat_1 = \concat_2 \in \form$,
    $\concat_2 \in \Output^*$ contains no variable,
    and $\concat_1 \in (\Output \cup \Var)^*$ contains
    at most one occurrence of each variable,
2) for all equations
    $\concat = \_$ and $\concat' = \_$ in $\form$,
    either $\concat$ and $\concat'$ do not have variables in common, 
    or $\project{\concat}{\Var} = \project{\concat'}{\Var}$,
    that is $\concat$ and $\concat'$ have the same sequence of variables.
We used the name \emph{sequential} due to this last fact.

\begin{example}
For $\x 1,\x 2,\x 3,\x 4,\x 5\in\Var$ and $p,q \in \Output^*$, each of the 
four formulas below is \regular{}:
\begin{align*}
\x 1 &= pq & \x 1 \x 3 &= qpqpqqpqpq \wedge \x 1 q \x 3 = qpqpqqqpqpq \\
\x 1 p \x 2 q \x 3 &= qppq & \x 1 pq \x 2 \x 3 &= pqpqpp \wedge \x 1 \x 2 qp \x 3 = pqppqp \wedge \x 5 p \x 4 = qpq
\end{align*}
The following formulas (and any formula containing them) are not \regular{}:
\begin{align*}
\x 1 pq \x 2 \x 3 = p \x 3 pq  &\text{\quad (rhs is not in $\Output^*$)} \\
\x 1 pq \x 2 p \x 3 \x 2 = ppqqpp  &\text{\quad ($\x 2$ appears twice in lhs)} \\
\x 1 pq \x 2 \x 3 = pqpqpp \wedge \x 2 p \x 5 = qpq &\text{\quad ($\x 2$ is shared)} \\
\x 1 pq \x 2 \x 3 = pqpqpp \wedge \x 1 p \x 3 \x 2 = pqppp &\text{\quad (different orderings of \x 1 \x 2 \x 3)}
\end{align*}
\end{example}

\newcommand{\pa}{\ensuremath{\text{p}}}
\newcommand{\pb}{\ensuremath{\text{q}}}
\newcommand{\ptr}{\ensuremath{\#}}
\begin{figure}
    \centering

    \begin{tikzpicture}[
        pnt/.style={->,circle, fill=black, inner sep=0pt, minimum size=4pt},
        pna/.style={circle,draw=black!80, line width=1pt, inner sep=0pt,minimum size=8pt},
        arr/.style={->,dashed, thick},
        edd/.style={->},
        inv/.style={draw=none,opacity=0},
        ver/.style={left=-3pt},
        hor/.style={above=-3pt},
        dia/.style={below left=-6pt},
        x=1.3cm,
        y=-1cm,
        scale=1
    ]
    \node[inv] at (-0.5,0) (S0) { };
    \node[pnt] at (0,0) (A0) { };
    \node[pnt] at (0,1) (A1) { };
    \node[pnt] at (0,2) (A2) { };
    \node[pnt] at (0,3) (A3) { };
    \node[inv] at (1,0) (B0) { };
    \node[pnt] at (1,1) (B1) { };
    \node[pnt] at (1,2) (B2) { };
    \node[pnt] at (1,3) (B3) { };
    \node[pnt] at (1,4) (B4) { };
    \node[inv] at (2,0) (C0) { };
    \node[pnt] at (2,1) (C1) { };
    \node[pnt] at (2,2) (C2) { };
    \node[pnt] at (2,3) (C3) { };
    \node[pna, double] at (2,4) (C4) { };
    
    \path[edd] (S0) edge node[inv] { } (A0);
   \path[edd] (A0) edge node[ver] { \pa{} } (A1);
   \path[edd] (A1) edge node[ver] { \pb } (A2);
   \path[edd] (A2) edge node[ver] { \pa } (A3);
   \path[edd] (B1) edge node[ver] { \pb } (B2);
   \path[edd] (B2) edge node[ver] { \pa } (B3);
   \path[edd] (B3) edge node[ver] { \pa } (B4);
   \path[edd] (C1) edge node[ver] { \pb } (C2);
   \path[edd] (C2) edge node[ver] { \pa } (C3);
   \path[edd] (C3) edge node[ver] { \pa } (C4);
   \path[edd] (A0) edge node[dia] { \ptr } (B1);
   \path[edd] (A2) edge node[dia] { \ptr } (B3);
   \path[edd] (A3) edge node[dia] { \ptr } (B4);
   \path[edd] (B1) edge node[hor] { \ptr } (C1);
   \path[edd] (B2) edge node[hor] { \ptr } (C2);
   \path[edd] (B3) edge node[hor] { \ptr } (C3);
   \path[edd] (B4) edge node[hor] { \ptr } (C4);
   \node[above=2mm] at (A0) { \x{0} };
   \node[above=2mm] at (B0) { \x{1} };
   \node[above=2mm] at (C0) { \x{2} };
   \node[below=0.5cm] at (B4) { \x{0} \pa{} \x{1} \x{2} = \pa\pb\pa\pa };
   \node[below=0.8cm] at (B4) { \phantom{d} };
   \end{tikzpicture}\hspace{0.6cm}
   \begin{tikzpicture}[
       pnt/.style={->,circle, fill=black, inner sep=0pt, minimum size=4pt},
       pna/.style={circle,draw=black!80, line width=1pt, inner sep=0pt,minimum size=8pt},
       arr/.style={->,dashed, very thick},
       edd/.style={->},
       inv/.style={draw=none,opacity=0},
       ver/.style={left=-3pt},
       hor/.style={above=-3pt},
       dia/.style={below left=-6pt},
       x=1.3cm,
       y=-1cm,
       scale=1
   ]
   \node[inv] at (-0.5,0) (S0) { };
   \node[pnt] at (0,0) (A0) { };
   \node[pnt] at (0,1) (A1) { };
   \node[pnt] at (0,2) (A2) { };
   \node[pnt] at (0,3) (A3) { };
   \node[inv] at (0,4) (A4) { };
   \node[pnt] at (1,0) (B0) { };
   \node[pnt] at (1,1) (B1) { };
   \node[pnt] at (1,2) (B2) { };
   \node[pnt] at (1,3) (B3) { };
   \node[inv] at (1,4) (B4) { };
   \node[inv] at (2,0) (C0) { };
   \node[inv] at (2,1) (C1) { };
   \node[pnt] at (2,2) (C2) { };
   \node[pnt] at (2,3) (C3) { };
   \node[pna, double] at (2,4) (C4) { };
   
   \path[edd] (S0) edge node[inv] { } (A0);
   \path[edd] (A0) edge node[ver] { \pb } (A1);
   \path[edd] (A1) edge node[ver] { \pa } (A2);
   \path[edd] (A2) edge node[ver] { \pa } (A3);
   \path[edd] (B0) edge node[ver] { \pb } (B1);
   \path[edd] (B1) edge node[ver] { \pa } (B2);
   \path[edd] (B2) edge node[ver] { \pa } (B3);
   \path[edd] (C2) edge node[ver] { \pa } (C3);
   \path[edd] (C3) edge node[ver] { \pa } (C4);
   \path[edd] (A0) edge node[hor] { \ptr } (B0);
   \path[edd] (A1) edge node[hor] { \ptr } (B1);
   \path[edd] (A2) edge node[hor] { \ptr } (B2);
   \path[edd] (A3) edge node[hor] { \ptr } (B3);
   \path[edd] (B1) edge node[dia] { \ptr } (C2);
   \path[edd] (B2) edge node[dia] { \ptr } (C3);
   \path[edd] (B3) edge node[dia] { \ptr } (C4);
    \node[above=2mm] at (A0) { \x{0} };
    \node[above=2mm] at (B0) { \x{1} };
    \node[above=2mm] at (C0) { \x{2} };
    \node[below=0.5cm] at (B4) { \x{0} \x{1} \pa{} \x{2} = \pb\pa\pa\pa };
    \node[below=0.8cm] at (B4) { \phantom{d} };
    \end{tikzpicture}\hspace{0.6cm}
    \begin{tikzpicture}[
        pnt/.style={->,circle, fill=black, inner sep=0pt, minimum size=4pt},
        pna/.style={circle,draw=black!80, line width=1pt, inner sep=0pt,minimum size=8pt},
        txt/.style={line width=0pt, inner sep=0pt,minimum size=0pt},
        arr/.style={->,dashed, very thick},
        edd/.style={->},
        inv/.style={draw=none,opacity=0},
        ver/.style={left=-3pt},
        hor/.style={above=-3pt},
        dia/.style={below left=-6pt},
        x=1.3cm,
        y=-1cm,
        scale=1
    ]
    \node[inv] at (-0.5,0) (S0) { };
    \node[pnt] at (0,0) (A0) { };
    \node[inv] at (0,1) (A1) { };
    \node[inv] at (0,2) (A2) { };
    \node[inv] at (0,3) (A3) { };
    \node[inv] at (0,4) (A4) { };
    \node[pnt] at (1,0) (B0) { };
    \node[pnt] at (1,1) (B1) { };
    \node[pnt] at (1,2) (B2) { };
    \node[pnt] at (1,3) (B3) { };
    \node[inv] at (1,4) (B4) { };
    \node[inv] at (2,0) (C0) { };
    \node[inv] at (2,1) (C1) { };
    \node[pnt] at (2,2) (C2) { };
    \node[pnt] at (2,3) (C3) { };
    \node[pna, double] at (2,4) (C4) { };
    
    \path[edd] (S0) edge node[inv] { } (A0);
    \path[edd] (B0) edge node[ver] { \pb } (B1);
    \path[edd] (B1) edge node[ver] { \pa } (B2);
    \path[edd] (B2) edge node[ver] { \pa } (B3);
    \path[edd] (C2) edge node[ver] { \pa } (C3);
    \path[edd] (C3) edge node[ver] { \pa } (C4);
    \path[edd] (A0) edge node[hor] { \ptr } (B0);
    \path[edd] (B1) edge node[dia] { \ptr } (C2);
    \path[edd] (B2) edge node[dia] { \ptr } (C3);
    \path[edd] (B3) edge node[dia] { \ptr } (C4);
        \node[above=2mm] at (A0) { \x{0} };
        \node[above=2mm] at (B0) { \x{1} };
        \node[above=2mm] at (C0) { \x{2} };
     \node[below=0.4cm] at (B4) { 
        $\quad \x{0} \x{1} \pa{} \x{2} = \pb\pa\pa\pa\, \wedge$};
     \node[below=0.8cm] at (B4) { 
        $\x{0} \pa{} \x{1} \x{2} = \pa\pb\pa\pa$
     };
        \end{tikzpicture}
    \caption{
        On the left, two automata representing the solutions of equations \x{0} \pa{} \x{1} \x{2} = \pa\pb\pa\pa{} and \x{0} \x{1} \pa{} \x{2} = \pb\pa\pa\pa{} respectively.
        On the right, their intersection represents the solutions of the
        conjunction of equations. Note that the third automaton can be
        obtained from the first (and the second) by removing states and 
        transitions.
    }    \label{figure:intersectionexample}
\end{figure}
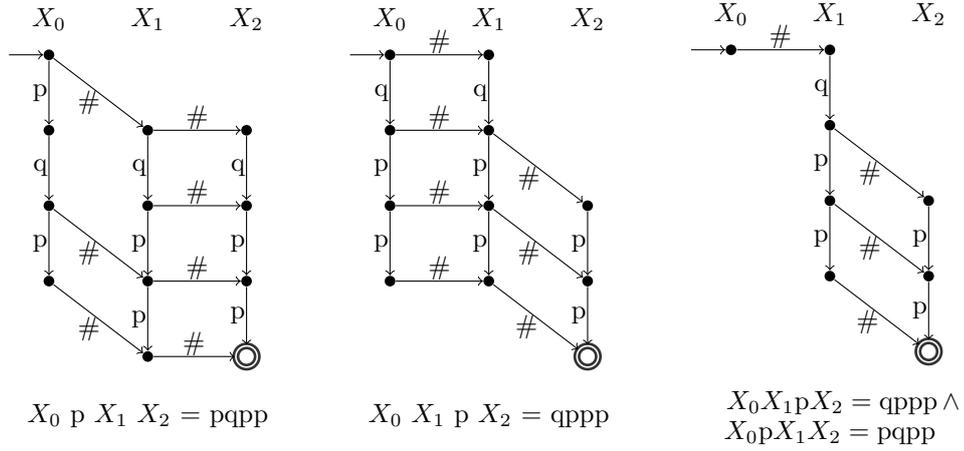
 
We prove that any \regular{} word formula $\form$ can be 
solved in polynomial time. 
    
\begin{restatable}{lem}{rplemma}
\label{lemma:rp}
Let $\form$ be a \regular{} word formula.
Let $n$ be the number of equations in $\form$,
$V$ the number of variables, and 
$C$ be the size of the largest constant appearing in $\form$.
We can determine in polynomial time $O(nVC)$ whether
$\form$ is satisfiable. When it is, we can also produce a satisfying
assignment for $\form$.
\end{restatable}

\begin{proof}
(Sketch)
We construct for each equation in 
$\form$ a DFA which represents succinctly all the possible assignments for 
this equation. Then, we take the intersection of all these DFAs, and 
obtain the possible assignments that satisfy all equations (i.e.~the 
assignments that satisfy formula $\form$).
The crucial part of the proof is to prove that this intersection can be 
computed in polynomial time, and does not produce an exponential blow-up
as can be the case with arbitrary DFAs. 
We prove this by carefully inspecting the DFAs representing the assignments,
and using the special form they have.
We show the intersection of two such DFAs $A$ and $B$ is a DFA whose size 
is smaller than both the sizes of $A$ and $B$ 
(instead of being the product of the sizes of $A$ and $B$, as can be the case 
for arbitrary DFAs). See Figure~\ref{figure:intersectionexample} for an
illustration of this intersection.
\end{proof}

\subsection{Algorithm for Learning from a Sample}

\begin{algorithm}
{\bf Input: } 
    A sample $\sample$ whose domain is closed under subtree.\\
{\bf Output: }
    If there exists a \stw{} $\trans$ such that $\trans(t) = w$ for all 
    $(t,w) \in \sample$, output $\yes$ and $\trans$,
    otherwise, output $\no$.
\begin{enumerate}
\item 
    Build the \regular{} formula 
        $\varphi \equiv \bigwedge_{(t,w) \in \sample} \getrequation(t,w,\sample)$
\item 
    Check whether $\varphi$ has a satisfying assignment $\morph$ as follows:
        (see Lemma~\ref{lemma:rp}):
    \begin{itemize}
    \item For every word equation 
    $\getrequation(t,w,\sample)$ where $t$ has root $f$, build a DFA
    that represents all possible solutions for the words
    $\morph(f,0)$,\dots,$\morph(f,k)$.
    \item 
        Check whether the intersection of all DFAs contains some word $w$.
        \begin{itemize}
        \item If no, exit the algorithm and return $\no$.
        \item 
            If yes, define the words $\morph(f,0)$,\dots,$\morph(f,k)$ 
            following $w$.
        \end{itemize}
    \end{itemize}
\item 
    Return ($\yes$ and) $\gettransducer{\morph}$.
\end{enumerate}

\caption{
    Learning \stw{s} from a sample.
}

\label{figure:algosample}
\label{algo}
\end{algorithm}
 
Consider a sample $\sample$ such that $\dom(\sample)$ is closed under subtree.
Given $(t,w) \in \sample$, we define the word equation $\getequation(t,w)$ as:
\[
  \default{\Sigma}(t) = w
\]
where the left hand side $\default{\Sigma}(t)$ is a concatenation of elements 
from $\SigmaPairs$, considered as word variables, and the right hand 
side $w \in \Output^*$ is considered to be a word constant.

Assume all equations corresponding to a set of input/output 
examples are simultaneously satisfiable, with an assignment 
$\morph: \SigmaPairs \to \Output^*$.
Our algorithm then returns the \stw{} $\trans = \gettransducer{\morph}$,
thus guarantying that $\trans(t) = w$ for all $(t,w) \in \Sigma$.

If the equations are not simultaneously satisfiable, 
our algorithm returns \no.

\begin{example}\label{example:equations}
For $\Sigma = \{ \nnil^{(0)}, \cons^{(2)}, \node^{(2)}, \ddiv^{(0)}, \dpre^{(0)}, \dspan^{(0)}\}$, given the examples:
\begin{align*}
\default{\Sigma}(\node(\ddiv,\nnil)) & = \text{``<.div''} \\
\default{\Sigma}(\ddiv) & = \text{``div''} &
\default{\Sigma}(\dspan) & = \text{``span''} &
\default{\Sigma}(\dpre) & = \text{``pre''} \\
\default{\Sigma}(\cons(\node(\ddiv,\nnil), \nnil) & = \text{``(<.div)''} &
\default{\Sigma}(\nnil) & = \text{``''}
\end{align*}
we obtain the following equations:
\begin{align*}
(\node, 0)\cdot (\ddiv, 0) \cdot (\node, 1) \cdot (\nnil, 0) \cdot (\node, 2) & = \text{``<.div''} \\
(\ddiv, 0) & = \text{``div''} \\
(\dspan, 0) & = \text{``span''} \\
(\dpre, 0) & = \text{``pre''} \\
(\cons, 0) \cdot (\node, 0)\cdot (\ddiv, 0) \cdot (\node, 1) \cdot (\nnil, 0) \cdot \\ (\node, 2) \cdot (\cons, 1) \cdot (\nnil, 0) \cdot (\cons, 2) & = \text{``(<.div)''} \\
(\nnil, 0) & = \text{``''}
\end{align*}

A satisfying assignment for these equations is the morphism 
$\getmorph{\tau}$ given in Example~\ref{example:morph}.
Note that this assignment is not unique (see Example~\ref{example:ambiguous}). 
We resolve ambiguities in Section~\ref{sec:learningwithoutambiguity}.
\end{example}

To check for satisfiability of $\bigwedge_{(t,w) \in \sample} \getequation(t,w)$,
we slightly transform the equations in order to obtain a \regular{} formula.
For $(t,w) \in \sample$, with $t = f(t_1,\dots,t_k)$,
we define the word equation $\getrequation(t,w,\sample)$ as:
\[
  (f,0)\, w_1\, (f,1) \cdots w_k\, (f,k) = w
\]
where for all $i \in \set{1,\dots,k}$, $w_i = \sample(t_i)$.
Note that $\sample(t_i)$ must be defined, since $t$ is in the domain of 
$\sample$, which is closed under subtree.
Moreover, the formula
\[
    \form \equiv \bigwedge_{(t,w) \in \sample} \getrequation(t,w,\sample)
\]
is satisfiable iff
\(
    \bigwedge_{(t,w) \in \sample} \getequation(t,w)
\)
is satisfiable.

Finally, $\form$ is a \regular{} formula.
Indeed, two equations corresponding to trees having
the same root $\ranked{f}{k} \in \Sigma$ 
have the same sequence of variables $(f,0) \dots (f,k)$
in their left hand sides.
And two equations corresponding to trees not having the same root 
have disjoint variables. Thus, using Lemma~\ref{lemma:rp}, we can check 
satisfiability of $\form$ in polynomial time (and obtain a satisfying 
assignment for $\form$ if there exists one).

\begin{thm}[Correctness and running time of Algorithm~\ref{algo}]
Let $\sample$ be a sample whose domain is closed under subtree.
If there exists a \stw{} $\trans$ such that $\trans(t) = w$ for all 
$(t,w) \in \sample$, Algorithm~\ref{algo} returns one such \stw{}.
Otherwise, Algorithm~\ref{algo} returns \no.
Algorithm~\ref{algo} terminates in time polynomial in the size of $\sample$.
\end{thm}

\begin{proof}

Assume $\form$ has a satisfying assignment 
$\morph: \SigmaPairs \to \Output^*$, in step (2) 
of Algorithm~\ref{algo}.
In that case, Algorithm~\ref{algo} returns 
$\trans = \gettransducer{\morph}$.
By definition of $\form$, we know, for all $(t,w) \in \sample$,
$\morph(\default{\Sigma}(t)) = w$.
Moreover, since $\getmorph{\trans} = \morph$, 
we have by Lemma~\ref{lemma:default} that 
$\trans(t) = \morph(\default{\Sigma}(t))$, so $\trans(t) = w$.

Conversely, if there exists $\trans$ such that 
$\trans(t) = w$ for all $(t,w) \in \sample$.
Then, again by Lemma~\ref{lemma:default}, 
$\getmorph{\trans}$ is a satisfying assignment for $\form$, and
Algorithm~\ref{algo} must return \yes.

The polynomial running time follows from Lemma~\ref{lemma:rp}.

\begin{remark}
For samples whose domains are not closed under subtree, we may modify 
Algorithm~\ref{algo} to check for satisfiability of word equations which
are not necessarily sequential. In that case, we are not guaranteed that 
the running time is polynomial.
\end{remark}
\end{proof}

\section{Learning \stw{s} Without Ambiguity}\label{sec:learningwithoutambiguity}

The issue with Algorithm~\ref{figure:algosample} is that
the \stw{} expected by the user may be different than the one
returned by the algorithm (see Example~\ref{example:ambiguous} below).
To circumvent this issue, we use the notion of \emph{tree test set}.
Formally, a set of trees $T \subseteq \domain$ 
is a \emph{\treetestset{} for the domain $\domain$} 
if for all \stw{s} $\trans_1$ and $\trans_2$, 
$\restrict{\semantics{\trans_1}}{T} = 
  \restrict{\semantics{\trans_2}}{T}$ implies 
$\restrict{\semantics{\trans_1}}{\domain} = 
  \restrict{\semantics{\trans_2}}{\domain}$.

\begin{example}~\label{example:ambiguous} The transducer $\trans_2$ defined 
below satisfies the requirements of Example~\ref{example:equations} but is 
different than the transducer in Example~\ref{example:transducer}. Namely, the 
values in the box have been switched.
\begin{align*}
\AssociateConstants_2(\node) =& (\text{``<.''}, \varepsilon, \varepsilon) \\
\AssociateConstants_2(\ddiv) =& (\text{``div''}) &
\AssociateConstants_2(\dpre) =& (\text{``pre''}) &
\AssociateConstants_2(\dspan) &= (\text{``span''}) \\
\AssociateConstants_2(\cons) =& (\text{``(''}, \boxed{\varepsilon,  \text{``)''}}) &
\AssociateConstants_2(\nnil) =& (\varepsilon)
\end{align*}
We can verify that the two transducers are not equal on the domain 
$\exampledomain$:
\begin{align*}
\trans(\cons(\node(\ddiv, \nnil), \cons(\node(\ddiv, \nnil), \nnil))) & = \text{``(<.div)(<.div)''}\\
\trans_2(\cons(\node(\ddiv, \nnil), \cons(\node(\ddiv, \nnil), \nnil))) & = \text{``(<.div(<.div))''}
\end{align*}
\end{example}

Therefore, if a user had the \stw{} $\trans$ in mind when giving the 
sample of Example~\ref{example:equations}, it is still possible 
that Algorithm~\ref{algo} returns $\trans_2$. However, by definition of 
\emph{tree test set}, if the sample given to Algorithm~\ref{algo}
contains a tree test set for $\exampledomain$, we are guaranteed that
the resulting transducer is equivalent to the transducer that the user has in 
mind, for all trees on $\exampledomain$.

Our goal in this section is to compute from a given domain $\domain$ a
tree test set for $\domain$. 
The notion of tree test set is derived from the well-known notion of 
\emph{test set} in formal languages.
The \emph{test set} of a language $L$ (a set of words) is a subset 
$T \subseteq L$ such that 
for any two morphisms $f,g: \Sigma^* \rightarrow \Gamma^*$, 
$\restrict{f}{T} = \restrict{g}{T}$ implies 
$\restrict{f}{L} = \restrict{g}{L}$.

To compute a tree test set $T$ for $\domain$, we first compute a test set 
$T_G$ for the context-free language 
$G = \set{\default{\Sigma}(t)\ |\ t \in \domain}$ (built in
Lemma~\ref{lemma:stwtomorphism}), and then define 
$T = \set{ \gettree(w)\ |\ w \in T_G}$.
We prove in Lemma~\ref{lemma:trees} that $T$ is indeed a tree test set for 
$\domain$.

We introduce in Section~\ref{section:cftestset} a new construction,
asymptotically optimal, for building test sets of context-free languages.
We show in Section~\ref{subsection:treetestset} how this translates to a
construction of a tree test set for a domain $\domain$. We also give a
sufficient condition of $\domain$ so that the obtained tree test set is 
closed under subtree.
This allows us to present, in Section~\ref{subsection:algo2}, an algorithm 
that learns \stw{s} from a domain $\domain$ in polynomial-time (by building
the tree test set $T$ of $\domain$, and asking to the user the outputs
corresponding to the trees of $T$).

\subsection{Test Sets for Context-Free Languages}

\label{sec:testsetscontextfree}
\label{section:cftestset}

We show in this section how to build, from a context-free grammar $G$,
a test set of size of $O(|G|^3)$. Our construction is asymptotically optimal.
We reuse lemmas from \cite{plandowski_testset_1994,plandowski_testset_1995},
which were originally used to give a $O(|G|^6)$ construction.

\subsubsection{Plandowski's Test Set}

The following lemma was originally used in 
\cite{plandowski_testset_1994,plandowski_testset_1995} to show that any 
linear context-free grammar has a test set containing
at most $O(|\Prod|^6)$ elements.
We show in Section~\ref{subsection:linear}
how this lemma can be used to show a $2|\Prod|^3$ bound.

Let $\Sigma_4 = 
    \set{
        a_i, 
        \closing{a_i},
        b_i,
        \closing{b_i}\ |\ i \in \set{1,2,3,4}}$
be an alphabet.
We define:
\begin{flalign*}
L_4 = 
    \{&x_4 \, x_3\, x_2 \, x_1 \, 
        \closing{x_1} \, 
        \closing{x_2} \, 
        \closing{x_3} \, 
        \closing{x_4}\ |\ \forall i \in \set{1,2,3,4}.\ 
            (x_i,\closing{x_i}) = (a_i,\closing{a_i}) \lor
            (x_i,\closing{x_i}) = (b_i,\closing{b_i})
\}
\end{flalign*}
and 
$T_4 = L_4 \setminus \set{b_4 \,b_3\,b_2\,b_1\,
    \closing{b_1}\,\closing{b_2}\,\closing{b_3}\,\closing{b_4}}$.

The sets $L_4, T_4 \subseteq \Sigma_4$ 
    have $16$ and $15$ elements respectively.

\begin{lem}[\cite{plandowski_testset_1994,plandowski_testset_1995}]
\label{lemma:t4l4}
$T_4$ is a test set for $L_4$.
\end{lem}

\subsubsection{Linear Context-Free Grammars}
\label{subsection:linear}

We now prove that for any linear context-free grammar $G$,
there exists a test set whose size is $2|\Prod|^3$.
Like the original proof of
\cite{plandowski_testset_1994,plandowski_testset_1995} that gave a 
$O(|\Prod|^6)$ upper bound, our proof 
relies on Lemma~\ref{lemma:t4l4}. However, our proof uses
a different construction to obtain the new, tight, bound.

\begin{restatable}{thm}{cftestset}
\label{theorem:cftestset}
Let $\gram = (\NTerm,\Term,\Prod,\Start)$ be a linear context-free grammar. 
There exists a test set $T \subseteq \gram$ for $\gram$ containing 
at most $2|\Prod|^3$ elements.
\end{restatable}

\begin{proof}
(Sketch)
Our proof relies on the fact that a linear grammar $G$ can be seen as a 
labelled graph whose nodes are non-terminals and whose transitions are rules of
the grammar. A special node labelled $\final$ is used for rules whose 
right-hand-sides are constant.
We define the notion of \emph{optimal path} in this graph. We use
optimal paths to define paths which are piecewise optimal. 
More precisely, for $k \in \Nat$, a word belongs to the set 
$\optimalset_k(G)$ if it can be derived in $G$ by a path that can be 
split into $k+1$ optimal paths.
We then prove that $\optimalset_3(G)$ forms a test set for $G$
(by using Lemma~\ref{lemma:t4l4}), which ends our proof as
$\optimalset_3(G)$ contains $O(|\Prod|^3)$ elements.

\end{proof}

We make use of this theorem in the next section to obtain test sets 
for context-free grammars which are not necessarily linear.

\subsubsection{Context-Free Grammars}

To obtain a test set for a context-free grammar $\gram$ which
is not necessarily linear, 
\cite{plandowski_testset_1994} constructs from $\gram$ a
linear context-free grammar, $\Lin{\gram}$, which produces a subset of $\gram$,
and which is a test set for $\gram$.

Formally, $\Lin{\gram}$ is derived from $\gram$ as follows:
\begin{itemize}
\item 
    For every productive non-terminal symbol $A$ in \gram{},
    choose a word $x_A$ produced by $A$.
\item 
    Every rule $r: A \to x_0 A_1 x_1 \ldots A_n x_n$ in $\gram$, where for
    every $i$, 
    $x_i \in \Term^*$ and $A_i \in \NTerm$ is productive, 
    is replaced
    by $n$ different rules, each one obtained from $r$ by 
    replacing all $A_i$ with $x_{A_i}$, except one.
\end{itemize}

Note that the definition of \Lin{\gram} is not unique, and depends on the 
choice of the words $x_A$. The following result holds  for any 
choice of the words $x_A$.

\begin{lem}[\cite{plandowski_testset_1994,plandowski_testset_1995}]
\label{lemma:ling}
$\Lin{\gram}$ is a test set for $\gram$.
\end{lem}
Using Theorem~\ref{theorem:cftestset}, we improve the 
$O(|G|^6)$ bound of \cite{plandowski_testset_1994,plandowski_testset_1995}
for the test set of $\gram$ to $2|\gram|^3$.

\begin{thm}
\label{theorem:cftestsetgen}
Let $\gram = (\NTerm,\Term,\Prod,\Start)$ be a context-free grammar. 
There exists a test set $T \subseteq \gram$ for $\gram$ containing 
at most $2|\gram|^3$ elements.
\end{thm}

\begin{proof}
Follows from Theorem~\ref{theorem:cftestset},
Lemma~\ref{lemma:ling}, and from the fact that 
$\Lin{\gram}$ has at most
$|\gram| = \sum_{\mkrule{\nterm}{\rhs} \in \Prod} (|\rhs| + 1)$
rules. (When constructing $\Lin{\gram}$, each rule 
$\mkrule{\nterm}{\rhs}$ of $\gram$ is duplicated at most 
$|\rhs|$ times.)
\end{proof}

\subsection{Tree Test Sets for Transducers}

\label{subsection:treetestset}

We use the results of the previous section to construct a tree test set for a 
domain $\domain$. 

\begin{lem}
\label{lemma:trees}
Any domain $\domain = (\Sigma,\State,I,\delta)$ has a \treetestset{} $T$
of size at most $O(|\domain|)^3$.
Moreover, if $I = Q$, then we can build $T$ such that $T$ is closed under 
subtree.
\end{lem}

\begin{proof}
Intuitively, we build the \treetestset{} for $\domain$ by taking 
the set of trees corresponding to the test set of $\gram$, where 
$\gram$ is the grammar built in Lemma~\ref{lemma:stwtomorphism}.

Let $\trans_1$ and $\trans_2$ be two \stw{s}.
Let $T_\gram$ be a test set for $\gram$.
Define $T = \set{ \gettree(w) \ |\ w \in T_\gram}$.
By Theorem~\ref{theorem:cftestsetgen}, we can assume
$T_\gram$ has size at most $|\gram|^3$, and hence, $T$ has size at most
$|\domain|^3$. 
Let $\morph_1$ and $\morph_2$ be $\getmorph{\trans_1}$
and $\getmorph{\trans_2}$, respectively.
We have:
\begin{flalign*}
&\restrict{\semantics{\trans_1}}{T} = 
  \restrict{\semantics{\trans_2}}{T} \iff \\
&\forall t \in T.\ \trans_1(t) = \trans_2(t) \iff \\
&\forall w \in T_\gram.\ \trans_1(\gettree(w)) = \trans_2(\gettree(w)) 
    \iff \text{ (by Lemma~\ref{lemma:default})}\\
&\forall w \in T_\gram.\ 
    \morph_1(\default{\Sigma}(\gettree(w))) =
    \morph_2(\default{\Sigma}(\gettree(w))) 
    \iff \text{ (by definition of $\gettree$)}\\
&\forall w \in T_\gram.\ 
    \morph_1(w) = \morph_2(w) 
    \iff \text{ (since $T_\gram$ is a test set for $\gram$)}\\
&\forall w \in \gram.\ 
    \morph_1(w) = \morph_2(w) 
    \iff \text{ (see Lemma~\ref{lemma:stwtomorphism})}\\
&\restrict{\semantics{\trans_1}}{\domain} = 
  \restrict{\semantics{\trans_2}}{\domain}
\end{flalign*}

This ends the proof that $T$ is a \treetestset{} for $\domain$.

We now show how to construct $T$ such that it is closed under subtree.
For every non-terminal $\nterm$ of $\gram$, we define the minimal word 
$w_\nterm$.
These words are built inductively, starting from the non-terminals which
have a rule whose right-hand-side is only made of terminals.
In the definition of $\Lin{\gram}$, we use these words when modifying 
the rules of $\gram$ into linear rules.

When then define $T_\gram$ as the test set of $\Lin{\gram}$ (which is 
also a test set of $\gram$), and
$T = \set{ \gettree(w) \ |\ w \in T_\gram} \cup 
    \set{ \gettree(w_\nterm) \ |\ \nterm \in \gram}$.
As shown previously, $T$ is a tree test set for $\domain$.
We can now prove that $T$ is closed under subtree.
Let $t = f(t_1, \ldots, t_k) \in T$. Let $i \in \set{1,\dots,k}$. 
We want to prove that $t_i \in T$.

We consider two cases.
Either there exists $w \in T_\gram$ such that 
$t = f(t_1, \ldots, t_k) = \gettree(w)$,
or there exists $\nterm \in \gram$, 
$t = f(t_1, \ldots, t_k) = \gettree(w_\nterm)$.

\begin{itemize}
\item 
First, if there exists $w \in T_\gram$ such that 
$t = f(t_1, \ldots, t_k) = \gettree(w)$.
Consider a derivation $p$ for $w$ in the $\Lin{\gram}$. 
By construction of $\Lin{\gram}$, the first rule is an $\varepsilon$-transition of 
the form $\mkrule{S}{N}$ while the second rule is of the form:
$$\mkrule{N}{(f, 0)\cdot w_1\cdot (f, 1)\cdots w_{j-1}\cdot(f, j-1)\cdot N_j\cdot (f, j)\cdot w_{j+1} \cdots w_k\cdot (f, k)}.$$
This second rule corresponds to a rule in $\gram$, of the form:
$$\mkrule{N}{(f, 0)\cdot N_1\cdot (f, 1)\cdots N_{j-1}\cdot(f, j-1)\cdot N_j\cdot (f, j)\cdot N_{j+1} \cdots N_k\cdot (f, k)}.$$

We then have two subcases to consider.
Either $i \neq j$, and in that case $t_i = \gettree(w_i)$.
By construction of $\Lin{\gram}$, $w_i$ must be equal to $w_\nterm$ for 
some $\nterm \in G$. Thus, we have $t_i \in T$ by definition of $T$.

Or $i = j$, in that case $t_i = \gettree(w')$, where 
$w'$ is derived by the derivation $p$ where the first two derivation rules,
outlined above, are replaced with the $\varepsilon$-rule 
$\mkrule{\Start}{N_i}$. This production rule is ensured to exist in 
$\Lin{\gram}$, as all states of $\domain$ are initial, so 
there exists a rule $\mkrule{\Start}{N_q}$ for all $q \in Q$.
(see definition of $\gram$ in Lemma~\ref{lemma:stwtomorphism}).
Then, since $w \in \optimalset_3(\Lin{\gram})$, and by construction of 
$\optimalset_3(\Lin{\gram})$, we conclude that 
$w' \in \optimalset_3(\Lin{\gram})$.
This ensures that $w' \in T_\gram$, and $t_i \in T$.

\item 
Otherwise, there exists $\nterm \in \gram$ such that 
$t = f(t_1,\dots,t_k) = \gettree(w_\nterm)$.
Using the fact that $w_\nterm$ was build inductively in the grammar $\gram$, 
using other minimal words $w_{\nterm'}$ for $\nterm' \in \gram$, 
we deduce there exists $\nterm' \in \gram$ such that 
$t_i = \gettree(w_{\nterm'})$, and $t_i \in T$.
\end{itemize}

\end{proof}

Lemma~\ref{lemma:lowerbound} shows the bound given in Lemma~\ref{lemma:trees} 
is tight, in the sense that there exists an infinite class of growing domains 
$\domain$ for which the smallest \treetestset{} has size $|\domain|^3$.

\begin{restatable}{lem}{lowerbound}
\label{lemma:lowerbound}
There exists a sequence of domains $\domain_1, \domain_2, \dots$
such that for every $n \geq 1$, 
the smallest \treetestset{} of $\domain_n$ has at least $n^3$ elements,
and the size of $\domain_n$ is linear in $n$.
Furthermore, this lower bound holds even with the extra assumption that 
all states of the domain are initial.
\end{restatable}

\begin{proof}
(Sketch)
Our proof is inspired by the lower bound proof for test sets of context-free
languages~\cite{plandowski_testset_1994,plandowski_testset_1995}. 
For $n \geq 1$, we build a particular domain $\domain_n$ (whose states are 
all initial), and we assume by 
contradiction that it has a test set $T$ of size less than $n^3$.
From this assumption, we expose a tree $t \in \domain_n$, as well as 
two \stw{s} $\trans_1$ and $\trans_2$ such that 
${\trans_1}_{|T} = {\trans_2}_{|T}$ but $\trans_1(t) \neq \trans_2(t)$.
\end{proof}

\subsection{Learning \stw{s} Without Ambiguity}

\label{subsection:algo2}

\begin{algorithm}
{\bf Input: } 
    A domain $\domain$, and an oracle \stw{} $\utrans$.\\
{\bf Output: } 
    A \stw{} $\trans$ functionally equivalent to $\utrans$.
\begin{enumerate}
\item 
    Build a tree test set $\set{t_1 \ldots t_n}$ of $\domain$, 
    following Lemma~\ref{lemma:trees}.
\item  \label{algline:askforoutputonebyone} 
    For every $t_i \in \{t_1 \ldots t_n\}$, ask the oracle for 
    $w_i = \utrans(t_i)$. 
\item
    Run Algorithm~\ref{algo} on the sample 
        $\set{ (t_i,w_i)\ |\ 1 \leq i \leq n}$.
\end{enumerate}

\caption{
    Learning \stw{s} from a domain.
}

\label{fig:activelearningalgorithm}
\label{algo2}
\end{algorithm}
 
Our second algorithm (see Algorithm~\ref{algo2})
takes as input a domain $\domain$, and computes 
a tree test set $T \subseteq \domain$.
It then asks the user the expected output for each tree $t \in T$.
The user is modelled by a \stw{} $\utrans$ that can be used 
as an oracle in the algorithm.
Algorithm~\ref{algo2} then runs Algorithm~\ref{figure:algosample} on the
obtained sample.
The \stw{} $\utrans$ expected by the user may still be syntactically 
different the \stw{} $\trans$ returned by our algorithm, but we are
guaranteed that 
$\restrict{\semantics{\trans}}{\domain} =
\restrict{\semantics{\utrans}}{\domain}$
(by definition of \treetestset).

\begin{thm}[Correctness and running time of Algorithm~\ref{algo2}]
Let $\utrans$ be a \stw{} (used as an oracle), and 
$\domain = (\Sigma,\State,I,\delta)$ a domain such that $I = Q$.
The output $\trans$ of Algorithm~\ref{algo2} is 
a \stw{} $\trans$ such that
$\restrict{\semantics{\trans}}{\domain} = 
    \restrict{\semantics{\utrans}}{\domain}$.

Furthermore,
Algorithm~\ref{algo2} invokes the oracle $O(|\domain|^3)$ times,
and terminates in time polynomial in $|\domain|$.
\end{thm}

\begin{proof}
The correctness of Algorithm~\ref{algo2} follows from the 
correctness of Algorithm~\ref{algo} and from the fact that $T$ is a 
\treetestset{} for $\domain$.
The fact that Algorithm~\ref{algo2} invokes the algorithm 
$O(|\domain|^3)$ times follows from the size of the \treetestset{}
(see Lemma~\ref{lemma:trees}).

Moreover, since all states of $\domain$ are initial, the tree test set of 
$\domain$ that we build is closed under subtree. The polynomial running time 
then follows from the fact that Algorithm~\ref{algo} ends in polynomial time 
for samples whose domains are closed under subtree.

\begin{remark}
Similarly to Algorithm~\ref{algo}, 
Algorithm~\ref{algo2} also applies for domains such that $I \neq Q$, but the 
running time is not guaranteed to be polynomial.
\end{remark}
\end{proof}

\section{Learning \stw{} Interactively}
\label{sec:learningstwinteractively}
\addtocounter{subsection}{1}

\begin{algorithm}
{\bf Input: } 
    A domain $\domain$, and an oracle \stw{} $\utrans$ 
        whose output alphabet is $\Output$.\\
{\bf Output: } 
    A \stw{} $\trans$ functionally equivalent to $\utrans$.
\begin{enumerate}
\item
    Initialize a map $\solutions$ from $\Sigma$ to Automata, such that 
    for $\ranked{f}{k} \in \Sigma$, $\solutions(f)$ recognizes 
    $\set{ x_0 \# \cdots \# x_k \ |\ x_i \in \Output^*}$,
\item 
    Build a tree test set $T$ of $\domain$, 
    following Lemma~\ref{lemma:trees}. 
\item
    Initialize a partial function  
    $\sample: \TreeA{\Sigma} \pto \Output^*$, initially undefined
    everywhere.
\item
    While $\dom(\sample) \neq T$:
    \begin{itemize}
    \item 
        Choose a tree $f(t_1,\dots,t_k) \notin \dom(\sample)$ such that all subtrees of $t$ belong to $\dom(\sample)$ 
        (possible since $T$ is closed under subtree).
    \item 
        Build the automaton $\alts$ recognizing 
        $\set{ 
            x_0\, \sample(t_1)\, x_1 \cdots \sample(t_k)\, x_k \ |\ 
                x_0 \# x_1 \cdots \# x_k \in \solutions(f)
         }$,
         representing all possibles values of $\utrans(t)$ that do not
         contradict previous outputs.
        \begin{itemize}
        \item 
            If $\alts$ recognizes only $1$ word $w$, define 
            $\sample(t) = w$.
        \item
            Otherwise ($\alts$ recognizes at least $2$ words), 
            define $\sample(t) = \utrans(t)$ using the oracle.
        \end{itemize}
    \item
        Update 
            $\solutions(f) = 
                \solutions(f) \cap 
                \getautomaton{t}{\sample(t)}$.
    \end{itemize}
\item
    Run Algorithm~\ref{algo} on $\sample$.
\end{enumerate}

\caption{
    Interactive learning of \stw{s}.
}

\label{fig:interactive}
\label{algo3}
\end{algorithm}
 
Our third algorithm (see Algorithm~\ref{algo3})
takes as input a domain $\domain$, and computes a tree test set 
$T \subseteq \domain$. For this algorithm, we require from the beginning
that all states of $\domain$ are initial, so that $T$ is closed under
subtree.
For a sample $\sample$ such that $\dom(\sample)$ is closed under subtree, and 
for $(t,w) \in \sample$, we denote by $\getautomaton{t}{w}$
the automaton $\getautomaton{\rp}{w}$
where $\rp = w$ is the equation $\getrequation(t,w,\sample)$.

Instead of building the sample $\sample$ and the intersection 
$\bigcap_{(t,w) \in \sample} \getautomaton{t}{w}$ all at once, 
like algorithms~\ref{algo} and \ref{algo2} do, Algorithm~\ref{algo3}
builds $\sample$ and the intersection incrementally.
It then uses the intermediary results to infer outputs, in order to avoid 
calling the oracle $\utrans$ too many times.
Overall, 
we prove that Algorithm~\ref{algo3} invokes the oracle $\utrans$ at most 
$O(|\domain|)$ times, while Algorithm~\ref{algo2} 
invokes it $O(|\domain|^3)$ times.

To infer outputs, Algorithm~\ref{algo3} maintains the following invariant for 
the while loop.
First $\sample$ is such that $\dom(\sample) \subseteq T$, and its domain
increases at each iteration.
Then, for any $\ranked{f}{k} \in \Sigma$,
$\solutions(f)$ is equal to 
$\bigcap_{(t,w) \in \sample} \getautomaton{t}{w}$, and thus 
recognizes the set
\[
    \set{ \morph(f,0) \# \morph(f,1) \# \dots \# \morph(f,k) \ |\ 
    \text{$\morph: \SigmaPairs \to \Output$ satisfies 
        $\bigwedge_{(t,w) \in \sample} \getrequation(t,w,\sample)$}}.
\]
Intuitively, $\solutions(f)$ represents the possible values for the output 
of $f$ in the transducer $\utrans$, based on the constraints given so far.

To infer the output of a tree $t = f(t_1,\dots,t_k)$, for some 
$\ranked{f}{k} \in \Sigma$, Algorithm~\ref{algo3} uses the fact that 
$\utrans(f(t_1,\dots,t_k))$ must be of the form 
$\morph(f,0) \sample(t_1) \morph(f,1) \cdots \sample(t_k) \morph(f,k)$
for some morphism $\morph: \SigmaPairs \to \Output$ satisfying 
        $\bigwedge_{(t,w) \in \sample} \getequation(t,w)$.
By construction, the NFA $\alts$, that recognizes the set
$\set{ 
            x_0\, \sample(t_1)\, x_1 \cdots \sample(t_k)\, x_k \ |\ 
                x_0 \# x_1 \cdots \# x_k \in \solutions(f)
         }$,
recognizes exactly these words of the form 
$\morph(f,0) \sample(t_1) \morph(f,1) \cdots \sample(t_k) \morph(f,k)$.

We then check whether $\alts$ recognizes exactly one word $w$,
in which case, we know $\utrans(t) = w$, and we do not need to invoke 
the oracle.
Otherwise, there are several alternatives which are consistent with the
previous outputs provided by the user, and we cannot infer $\utrans(t)$.
We thus invoke the oracle (the user) to obtain $\utrans(t)$.

Before proving the theorem corresponding to Algorithm~\ref{algo3},
we give a lemma on words which we use extensively in the theorem.

\begin{lem}
\label{lemma:commute}
Let $u,v,w \in \Output^*$.
If $uv = vu$ and $uw = wu$ and $u \neq \varepsilon$,
then $vw = wv$.
\end{lem} 

\begin{proof}
A word $p \in \Output^*$ is \emph{primitive} if there does not
exist $r \in \Output^*$, $i > 1$ such that $p = r^i$.
Proposition 1.3.2 of~\cite{lothaire_combinatorics_1997} states that the set of
words commuting with a non-empty word $u$ is a monoid generated by a single 
primitive word $p$. Since $v$ and $w$ both commute with $u$, there exist $i$ 
and $j$ such that $v = p^i$ and $w = p^j$, thus $vw = wv = p^{i+j}$.
\end{proof}

The difficult part of Theorem~\ref{theorem:linearinteraction} is to show 
the number of times the oracle $\utrans$ is invoked is $O(|\domain|)$.
We prove this by assuming by contradiction that the 
number of times $\utrans$ is invoked is strictly greater than 
$3|\domain| + |\State|$ times. We prove this entails there are four 
trees which are nearly identical and for which our algorithm invokes the oracle
(the four trees have the same root, and differ only for one child).
Then, by a close analysis of the word equations corresponding to these four 
terms, we obtain a contradiction by proving our algorithm must have been able 
to infer the output for at least one of those terms.

\begin{restatable}[Correctness and running time of Algorithm~\ref{algo3}]{thm}{linearinteraction}
\label{theorem:linearinteraction}
Let $\utrans$ be a \stw{} (used as an oracle), and 
$\domain = (\Sigma,\State,I,\delta)$ a domain such that $I = Q$.
The output $\trans$ of Algorithm~\ref{algo3} is 
a \stw{} $\trans$ such that
$\restrict{\semantics{\trans}}{\domain} = 
    \restrict{\semantics{\utrans}}{\domain}$.
    
Algorithm~\ref{algo3} ends in time polynomial in $|\domain|$ and the number of 
times it invokes the oracle $\utrans$ is in $O(|\domain|)$.
\end{restatable}

\begin{proof}
(Sketch) 
The correctness and the polynomial running time of Algorithm~\ref{algo3} can
be proved similarly to Algorithm~\ref{algo2}.
Note that we can check whether the NFA $A$ recognizes exactly one word. For that, we obtain a word $w$ that $A$ recognizes, and we intersect $A$ with the complement of an automaton recognizing $w$.

The crucial part of Algorithm~\ref{algo3} is that it invokes the oracle 
$\utrans$ at most $O(|\domain|)$ times. More precisely, we show that 
Algorithm~\ref{algo3} invokes $\utrans$ at most 
$|Q| + 3 \sum_{(q,\ranked{f}{k},(q_1,\dots,q_k) \in \delta} 1 + k$ times,
which is $|Q| + 3|\domain|$, and in $O(|\domain|)$.

The main goal is to prove that 
for any trees four trees of the same root ($t_a$, $t_b$, $t_c$, $t_d$) 
differing from only one their $i$th subtree 
(respectively $t_i^a$, $t_i^b$, $t_i^c$, $t_i^d$),
if we know the output of $\utrans$ on all subtrees of 
$t_a$, $t_b$, $t_c$, $t_d$, then we can infer the output for at least 
one of $t_a$, $t_b$, $t_c$, $t_d$ based on the previous outputs.
Let $x_i^l = \utrans(t_i^l)$ be the already known outputs of the sub-trees and $w_l = \utrans(t_l)$ the outputs to ask to the user, for $l \in \set{a,b,c,d}$. We obtain the following equations where $u, v$ represent the parts which do not change:
$$
w_a = u x_i^a v \quad
w_b = u x_i^b v \quad
w_c = u x_i^c v \quad
w_d = u x_i^d v
$$
We prove by contradiction that we could not have asked the user for all $w_l$ for $l \in \set{a,b,c,d}$, because at least one of the answer can be inferred
from the previous ones.
Here we illustrate two representative cases of the proof.

(1) One case is when $x_i^a$ and $x_i^b$ are neither prefix nor suffix of each 
other. By observing where $w_a$ and $w_b$ differ, we can recover $u$ and $v$,
and the algorithm could have inferred $w_c$ and $w_d$.

(2) Another case is when $x_i^a$, $x_i^b$, and $x_i^c$ are respectively of 
the form $x_1$, $x_1 x_2$ and $x_1 x_2 x_3$ for some 
$x_1,x_2,x_3 \in \Output^*$ with $x_2x_3 = x_3x_2$, and 
$x_2 \neq \varepsilon$, $x_3 \neq \varepsilon$. 
Since we asked the output $w_a$, $w_b$ and $w_c$, 
then after the first two questions, the values of $u$ and $v$ could 
not be determined. 
In particular, this means that there are some $u, v$ and $u', v'$ such that:  
$u x_1 v = u' x_1 v'$ and 
$u x_1 x_2 v = u' x_1 x_2 v'$ but 
$u x_1 x_2 x_3 v \neq u' x_1 x_2 x_3 v'$. 

By assuming without loss of generality that $u = u' u''$ and $v' = v''v$, 
we obtain that 
$u'' x_1 = x_1 v''$ and 
$u'' x_1 x_2 = x_1 x_2 v''$, 
thus $v'' x_2 = x_2 v''$, 
and then $x_2$ commutes with $v''$.
Since $x_2$ also commutes with $x_3$, we deduce $v''$ commutes with $x_3$, and 
then $u'' x_1 x_2 x_3 = x_1 x_2 x_3 v''$, which is a contradiction.
\end{proof}

\section{Tree with Values}\label{sec:treewithvalues}
\addtocounter{subsection}{1}

Until now, we have considered a set of trees $\TreeA{\Sigma}$ which contained only other trees as subtrees, and with a test set of size $O(n^3)$, although we have a linear learning time if we have interactivity.
However, in practice, data structures such as XML are usually trees containing \textit{values}. Values are typically of type \symbolstring or \symbolint{}, and may be used instead of subtrees.
For convenience, we will suppose that we only have \symbolstring{} elements, and that \symbolstring{} elements are rendered \textit{raw}.
We will demonstrate how we can directly obtain a test set of size $O(n)$.

\newcommand{\dstring}[1]{\ensuremath{v_{#1}}}
Formally, let us add a special symbol $\dstring{} \in \Sigma$, of arity 0, which has another version which can have a parameter.
For each string $s\in\Output^*$ we can thus define the symbol $\dstring{s}$ and extend the notion of trees and domains as follows.

For a set of trees $\TreeA{}$, we define the extended set $\TreeA{}'$ by:
$$\TreeA{}' = \{t'\ |\ \exists t\in \TreeA{}, t'\text{ is obtained from }t\text{ by replacing each }\dstring{}\text{ by a }\dstring{s}\text{ for some }s\in\Output^* \}$$

Note that given a domain $\domain$ and a height $h$, there is an infinite number 
of trees of height $h$ in $\domain'$, while only a finite number in $\domain$.
Fortunately, thanks to the semantics of the transducers on \dstring{s} we 
define below, finding the tree test sets is easier in this setting.

For any transducer $\trans$ we extend the definition of $\semantics{\trans}$ to $\TreeA{\Sigma}'$ by defining $\semantics{\trans}(\dstring{s}) = s$.
We naturally extend the definition of tree test set of an extended domain $\domain'$ to be a set $T' \subset \domain'$ such that for all \stw{s} $\trans_1$ and $\trans_2$, 
$\restrict{\semantics{\trans_1}}{T'} = 
  \restrict{\semantics{\trans_2}}{T'}$ implies 
$\restrict{\semantics{\trans_1}}{\domain'} = 
  \restrict{\semantics{\trans_2}}{\domain'}$.
After proving the following lemma, we will state and prove the theorem on linear test sets.

\begin{lem}\label{lemma:annihilation}
For $a, b, x, y \in \Output^*$, $c \neq d$ in $\Output$, if $a c x = b c y$ and $a d x = b d y$, then $a = b$.
\end{lem}
\begin{proof}
Either $a$ or $b$ is a prefix of the other. Let us suppose that $a = b k$ for some suffix $k \in \Output^*$. It follows that $k c x = c y$ and $k d x = d y$. If $k$ is not empty, then $k$ starts with $c$ and with $d$, which is not possible. Hence $k$ is empty and $a = b$.
\end{proof}

\newcommand{\ta}{\ensuremath{\text{``\#''}}}
\newcommand{\tb}{\ensuremath{\text{``?''}}}
\begin{restatable}{thm}{theoremwithvalues}
\label{theorem:withvalues}
If the domain $\domain = (\Sigma, Q, I, \delta)$ is such that for every $f \in \Sigma$ of arity $k > 0$, there exist trees in $t_1,\ldots,t_k \in \domain$ such that $f(t_1,\ldots,t_k) \in \domain$ and each $t_i$ contains at least one $\dstring{}$,
then there exists a tree test set of $\domain'$ of linear size $O(|\Sigma|\cdot A)$ where $A$ is the maximal arity of a symbol of $\Sigma$.
\end{restatable}
\newcommand{\exts}{\ensuremath{\text{T}'}}
\newcommand{\treec}[3]{\ensuremath{{#1}^{#2}\!\!\left[{#3}\right]}}
\begin{proof}
(Intuition) 
Using the trees provided in the theorem's hypothesis, we build a
linear set of trees of $\domain'$ where the $\dstring{}$ nodes are replaced 
successively by two different symbols $\dstring{\ta}$ and $\dstring{\tb}$. 
Then, we prove that any two \stw{s} which are equal on this set of trees, 
are syntactically equal. 
\end{proof}

\section{Implementation}\label{sec:implementation}

Our tool (walkthrough in Section~\ref{section:walkthrough}) is open-source
and available at \url{https://github.com/epfl-lara/prosy}.
It takes as input an ADT represented by case class definitions written in a 
Scala-like syntax, and outputs a recursive printer for this ADT.
For the automata constructions of Algorithm~\ref{algo3}, we used the 
{\tt brics} Java library\footnote{http://www.brics.dk/automaton/}.

In the walkthrough, notice that our tool gives propositions to the user so 
that the user does not have to enter the answers manually.
The user may choose how many propositions are to be displayed (default is 9).
To obtain these propositions, we use the following procedure.
Remember that for each tree $t$ for which we need to obtain the output, 
Algorithm~\ref{algo3} builds an NFA $\alts$ that recognizes the set of all 
possible outputs for $t$ (see Section~\ref{sec:learningstwinteractively}).
We check for the existence of an accepted word $w_0$ in $\alts$, and
compute the intersection $A_1$ between $\alts$ and an automaton recognizing 
all words except $w_0$. We then have two cases. Either $A_1$ is empty, and 
therefore we know the output for tree $t$ is $w_0$. In that case, we do not 
need to interact with the user, and can continue on to the next tree. 
Otherwise, $A_1$ recognizes some word $w_1 \neq w_0$, which we display as a 
proposition to the user (alongside $w_0$). We then obtain $A_2$ as the 
intersection between $A$ and an automaton recognizing all words except 
$w_0$ and $w_1$. We continue this procedure until we have $9$ propositions
(or whichever number the user entered), or when the intersected automaton 
becomes empty.

Concerning support for the String data type, 
we use ideas from Section~\ref{sec:treewithvalues} and reused our 
code from Algorithm~\ref{algo3} to infer outputs.
Technically, we replace the String data type with an abstract class with 
two case classes, \textsf{foo}, and \textsf{bar}, that must be printed as 
``foo'' and ``bar'' respectively.
We then obtain an ADT without Strings, on which we apply the 
implementation of Algorithm~\ref{algo3} described above. We handle the 
Int and Boolean data types similarly, each with two different values \textit{which are not prefix of each other}~(we refer to the proof of theorem~\ref{theorem:linearinteraction}).
\section{Evaluation}\label{sec:evaluation}
Although this work is mostly theoretical, we now depict through some benchmarks 
how many and which kind of questions our system is able to ask 
(Figure \ref{fig:evaluation}).

\begin{figure}
\begin{tabular}{lcccccc}
& Test set & \multicolumn{5}{c}{The output was} \\
\cline{3-7}
& size & inferred & asked & \multicolumn{3}{c}{asked with\ldots} \\
\cline{5-7}
Name & \textit{Total} & \textit{total} & \textit{total} & nothing & 
a hint & suggestions \\
\hline\hline
Grammar (Sec. \ref{sec:walkthrough}) &
                            116   & 102   & 14  & 6   & 6   & 2 \\
\hline
Html tags (Ex. \ref{example:transducer}, \ref{example:equations}, \ref{example:ambiguous})  &  35   & 28    & 7   & 4   & 2   & 1 \\
Html tags+attributes      &  60   & 52    & 8   & 2   & 4   & 2 \\
Html xml+attributes       &  193  & 179   & 14  & 5   & 3   & 6 \\
\hline
Binary (01001x)   &  15   & 12    & 3   & 1   & 2   & 0 \\
Binary (11x)      &  15   & 12    & 3   & 3   & 0   & 0 \\
Binary (ababx)    &  15   & 11    & 4   & 3   & 0   & 1 \\
Binary (01001)    &  15   & 10    & 5   & 3   & 0   & 2 \\
Binary (aabababbab)& 15   & 9     & 6   & 3   & 0   & 3 \\
\hline
$A_x(B_y(F_z))$ 1         & 3     & 0     & 3   & 1   & 2   & 0 \\
$A_x(B_y(F_z))$ 2         & 14    & 8     & 6   & 3   & 3   & 0 \\
$A_x(B_y(F_z))$ 4         & 84    & 67    & 17  & 8   & 4   & 5 \\
$A_x(B_y(F_z))$ 8         & 584   & 552   & 32  & 19  & 5   & 8 \\
$A_x(B_y(F_z))$ 16        & 4368  & 4305  & 63  & 32  & 16  & 15
\end{tabular}
\caption{Comparison of the number of questions asked for different benchmarks.\label{fig:evaluation}}
\end{figure}

The first column is the name of the benchmark. The first two appear in Section \ref{sec:walkthrough} and in the examples. The third is a variation of the second where we add attributes as well, rendered ``\^{}.foo := "bar"''. The fourth is the same but rendered in XML instead of tags. Note that because we do not support duplication, we need to have a finite number of tags for XML.

The four rows ``binary'' illustrate how the number and type of questions may vary only depending on the user's answers.
We represent binary numbers as either \texttt{Empty} or \texttt{Zero(x)} or \texttt{One(x)} where \texttt{x} is a binary number.
We put in parenthesis what a user willing to print \texttt{Zero(One(Zero(Zero(One(Empty)))))} would have in mind.
The second and the third ``discard'' \texttt{Zero} when printing.
The fourth one prints \texttt{Empty} as empty, \texttt{Zero(x)} as \texttt{\{x\}ab} and \texttt{One(x)} as \texttt{a\{x\}b},
which result in an ambiguity not resolved until asking a 3-digit number.

The last five rows of Figure~\ref{fig:evaluation} also illustrate how the number of asked questions grows linearly, whereas the number of elements in the test set grows cubically. These five rows represent a set of classes of type A taking as argument a class of type B, which themselves take as argument a class of type F. We report on the statistics by varying the number of concrete classes between 1, 2, 4, 8 and 16 (see proof of Lemma~\ref{lemma:lowerbound})

The second column is the size of the test set. For the last five rows, 
the test set contains a cubic number of elements.
The third column is the number of answers our tool was able to ``infer'' based on 
previously ``asked'' questions, whose total number is in the fourth column. 
The fourth column plus the third one thus equal the second one.

Columns five, six and seven decompose the fourth column into the questions which 
were either asked 
without any indication, 
or with a hint of type ``[...]foo[...]'' (because the arguments were known),
or with explicit suggestions where the user just had to enter a number for the 
choice (see Section~\ref{sec:implementation}).
\section{Related Work}

Our approach of proactively learning transducers by example, or tree-to-string programs, can be viewed as a particular case of Programming-by-Example.
Programming-by-example, also named inductive programming~\cite{polozov_flashmeta:_2015} or test-driven synthesis~\cite{perelman_test-driven_2014}, is gaining more and more attention, notably thanks to Flash Fill in Excel 2013~\cite{gulwani_synthesis_2012}. Subsequent work demonstrated that these techniques could widely be applicable not only to strings, but when extracting documents~\cite{le_flashextract:_2014},
normalizing text~\cite{kini_flashnormalize:_????} and 
number transformations~\cite{singh_synthesizing_2012}.
However, most state-of-the-art programming-by-example techniques rely on the fact that examples are unambiguous and/or that the example provider can check the validity of the final program~\cite{angluin_learning_1987}~\cite{yessenov_colorful_2013}~\cite{feser_synthesizing_2015}.
The scope of their algorithms may be larger but they do not guarantee formal result such as polynomial time or non-ambiguity, and often require the user to come up with the examples by himself.
More generally, synthesizing recursive functions has recently gained an interest among computer scientists from repairing fragments~\cite{koukoutos_update_2016} to very precise types~\cite{polikarpova_program_2016}, even by formalizing programming-by-example ~\cite{frankle_example_2016}.

Recently, research has pointed out that solving ambiguities is a key to make programming by example accessible, trustful and reduce the number of errors~\cite{mayer_user_2015}\cite{hottelier_programming_2014}. The power of interaction is already well known in more statistical approaches, e.g. machine learning~\cite{zhang_active_2015},
although recent machine-learning based formatting techniques could benefit from more interaction, because they acknowledge some anomalies~\cite{parr_towards_2016}.
In%
~\cite{gvero_interactive_2011}
and even%
~\cite{gvero_complete_2013}, the authors solve ambiguities by presenting different code snippets, obtained from synthesizing expressions of an expected type and from other sources of information. Nonetheless, the user has to choose between hard-to-read \textit{code snippets}. Instead of asking which transducer is correct, we ask for what is the right output.
Asking sub-examples at run-time proved to be a successful strategy when synthesizing recursive functions~\cite{albarghouthi_recursive_2013}.
To deal with ambiguous samples, they developed a \textsc{Saturate} rule to ask for inputs covering the inferred program.
In our case, however,
such coverage rule still yield the ambiguity raised in example~\ref{example:ambiguous}, leaving the chance of finding the right program to heuristics.

Researchers have investigated
fundamental properties of tree-to-string or tree-to-word transducers \cite{alur_expressiveness_2010}, including expressiveness of even more complex classes than we consider~\cite{alur_streaming_2012}, but none of them proposed a practical learning algorithm for such transducers.
The situation is analogous for Macro Tree Transducers~\cite{bahr_programming_2013}~\cite{engelfriet_output_2002}.
Lemay~\cite{lemay_learning_2010} explores the synthesis of top-down tree-to-tree transducers using an algorithm similar to $L^*$ for automata~\cite{angluin_learning_1987} and tree automata~\cite{besombes_learning_2004}.
These learning algorithms require the user to be in possession of a set of examples that uniquely defines the top-down tree transducer.
We instead are able to incrementally ask for examples which resolve ambiguities, although our transducers are single-state.
There are also probabilistic tree-to-string transducers~\cite{graehl_training_2004}, but they require the use of a corpus and are not adapted to synthesizing small-size code portions with a few examples.

A Gold-style learning algorithm~\cite{laurence_learning_2014,laurence_phd_2014,lemay_learning_2010}
was created for sequential tree-to-string transducers.
It runs in polynomial-time, but has a drawback:
it requires the input/output examples to form a 
\emph{characteristic sample} for the transducer which is 
being learned.
The transducer which is being learned is however not known in advance.
As such, it is not clear in practice how to construct such a characteristic
sample.
When the input/output examples do not form a characteristic sample, 
the algorithm might fail, and the user of the algorithm has no indication
on which input/output examples should be added to obtain a characteristic
sample.

In the case when trees to be printed are programming abstract syntax trees, our work is the dual of the mixfix parsing problem~\cite{jouannaud_programming_1992}. 
Mixfix parsing takes strings to parse and the wrapping constants to print the trees, and produces the shape of the tree for each string.
Our approach requires the shape of the trees and strings of some trees, and produces the wrapping constants to print the trees.

\subsection{Equivalence of top-down tree-to-string transducers}

Since tree test sets uniquely define the behavior of tree-to-string transducers, they can be used for checking tree-to-word transducers equivalence.
Checking equivalence of sequential (order-preserving, non-duplicating)
tree-to-string transducers can already be solved in polynomial 
time~\cite{staworko_printers_2009}, even when they are duplicating, and not necessarily order-preserving~\cite{maneth_xml_2007}.

It was also shown~\cite{helmut_seidl_equivalence_2015}
that checking equivalence of 
deterministic top-down macro tree-to-string 
transducers (duplication is allowed, storing strings in registers
to output them later is allowed) is decidable.
Complexity-wise, this result gives a co-randomized polynomial time algorithm 
for linear (non-duplicating) tree-to-string transducers.
This complexity result was recently improved in \cite{boiret_linear_2016},
where it was proved that checking equivalence of linear 
tree-to-string transducers can be done in polynomial time.

\subsection{Test sets}

The polynomial time algorithms of
\cite{staworko_printers_2009,boiret_linear_2016}
exploit a connection between the problem of checking equivalence 
of sequential top-down tree-to-string transducers and the problem
of checking equivalence of morphisms over context-free 
languages~\cite{staworko_printers_2009}.

This latter problem was shown to be solvable in polynomial 
time~\cite{plandowski_testset_1994,plandowski_testset_1995}
using test sets.
More specifically, 
this work shows that each context-free language $L$ has a (finite) test set 
whose size is $O(n^6)$ (originally ``finite'' in \cite{albert_ehrenfeucht_1985,guba_ehrenfeucht_1986}
and then ``exponential'' in \cite{albert_testset_1982}), where $n$ is the size of the grammar.
They also provide a lower bound on the sizes of the test sets of 
context-free languages, by exposing a family of grammars for which the size of 
the smallest test is $O(n^3)$.

As a result, when checking the equivalence of two morphisms $f$ and $g$
over a context-free language $L$, it is enough to check the equivalence on 
the test set of $L$ whose size is polynomial.
This result translates (as described in \cite{staworko_printers_2009})
to checking equivalence between sequential top-down tree-to-string transducers 
in the following sense. When checking the equivalence of two such 
transducers $P_1$ and $P_2$, it is enough to do so for a finite 
number of trees, which correspond to the test set of a particular context-free 
language. This language can be constructed from $P_1$ and $P_2$ 
in time $|P_1||P_2|$.

\begin{remark}
Theorem~\ref{theorem:cftestset} also helps improve the bound for checking
equivalence of $\stw$ with states, using the known
reduction from equivalence of $\stw$ with states to morphisms
equivalence over a context-free language (reduction similar to 
Lemma~\ref{lemma:stwtomorphism}, 
see \cite{staworko_printers_2009,laurence_phd_2014}).
\end{remark}

\section{Conclusion}

We have presented a synthesis algorithms that can learn from examples tree-to-string functions with the input tree as the only argument. This includes functions such as pretty printers. Crucially, our algorithm can
automatically construct a sufficient finite set of input trees, resulting in an interactive synthesis approach
that in which the user needs to answer only a linear number of questions in the grammar
size. Furthermore, the interaction process driven by our algorithm guarantees that there is no ambiguity: the recursive function of the expected form is unique for a given set of input-output examples.
Moreover, we have analyzed the structure of word equations that the algorithm needs to solve and
shown that they have a special structure allowing them to be solved in deterministic polynomial time, which results
in overall polynomial running time of our synthesizer.
Our results make a case that providing tests for tree-to-string functions is a viable alternative to writing the recursive programs directly, an alternative that is particularly appealing for non-expert users.
\bibliography{references} 

\appendix
\section{Injectivity of $\default{\Sigma}$}

\injective*

\begin{proof}
Assume that $\semantics{\default{\Sigma}}$ is not injective,
and let 
$t = f(t_1,\dots,t_n)$ and 
$t'= f'(t_1',\dots,t_m')$ be two trees with 
$t \neq t'$, such that 
$\default{\Sigma}(t) =
\default{\Sigma}(t')$.
We pick $t$ and $t'$ satisfying those conditions such that
$\default{\Sigma}(t)$ has the smallest possible
length.

By definition of $\default{\Sigma}$, we have
\[
    \default{\Sigma}(t) = (f,0) 
    \default{\Sigma}(t_1) (f,1) \cdots
    \default{\Sigma}(t_n) (f,n)
\]
and
\[
    \default{\Sigma}(t') = (f',0) 
    \default{\Sigma}(t_1') (f',1) \cdots
    \default{\Sigma}(t_m') (f',m).
\]

Since 
$\default{\Sigma}(t) =
\default{\Sigma}(t')$, we deduce
that $(f,0) = (f',0)$ and $f = f'$, meaning that $t$
and $t'$ have the same root.

Thus, $t = f(t_1,\dots,t_n)$, and $t'= f(t_1',\dots,t_n')$.

To conclude, we consider two cases.
If for all $i \in \set{1,\dots,n}$, 
$\default{\Sigma}(t_i) = \default{\Sigma}(t_i')$, 
we have $t_i = t_i'$, as the length of
$\default{\Sigma}(t_i)$ is strictly
smaller than $\default{\Sigma}(t)$.
This ensures that $t = t'$, and we obtain a contradiction.

On the other hand, if there exists 
$i \in \set{1,\dots,n}$,
$\default{\Sigma}(t_i) \neq \default{\Sigma}(t_i')$,
consider the smallest such $i$.
We then have:
\[
    \default{\Sigma}(t) = (f,0) 
    \default{\Sigma}(t_1) (f,1) \cdots
    \default{\Sigma}(t_{i-1}) (f,i-1) 
    \default{\Sigma}(t_i) (f,i) \cdots
    \default{\Sigma}(t_n) (f,n)
\]
and
\[
    \default{\Sigma}(t') = (f,0) 
    \default{\Sigma}(t_1) (f,1) \cdots
    \default{\Sigma}(t_{i-1}) (f,i-1) 
    \default{\Sigma}(t_i') (f,i) \cdots
    \default{\Sigma}(t_n') (f,n).
\]

Since $\default{\Sigma}(t) =
\default{\Sigma}(t')$, and the prefixes
are identical up until $(f,i-1)$, we deduce
\[
\default{\Sigma}(t_i) (f,i) \cdots
\default{\Sigma}(t_n) (f,n) = 
\default{\Sigma}(t_i') (f,i) \cdots
\default{\Sigma}(t_n') (f,n).
\]

We finally consider three subcases, with respect to this
last equation.
\begin{itemize}
\item If $\default{\Sigma}(t_i)$ and 
$\default{\Sigma}(t_i')$ have the same
length, we deduce 
$\default{\Sigma}(t_i) = \default{\Sigma}(t_i')$, contradicting 
our assumption.
\item If $\default{\Sigma}(t_i)$ is strictly
shorter than $\default{\Sigma}(t_i')$, we
deduce that
$\default{\Sigma}(t_i) (f,i)$ is a prefix
of $\default{\Sigma}(t_i')$. This is not 
possible as 
$\default{\Sigma}(t_i')$ must be well 
parenthesized 
if $(f,0)$ is seen as an open 
parenthesis, and $(f,i)$ as a closing parenthesis
(by definition of $\default{\Sigma}$).

\item The case where $\default{\Sigma}(t_i)$
is strictly longer than $\default{\Sigma}(t_i')$
is symmetrical to the previous one.
\end{itemize}

\end{proof}

\section{Proof of NP-completeness}\label{proof:npcomplete}

\npcomplete*

\begin{proof}

In general, we can check for the existence of $\trans$ in $\NP$
using the following idea. 
Every input/output example from the sample gives constraints on the 
constants of $\trans$.
Therefore, to check for the existence of $\trans$, it is sufficient to 
non-deterministically guess constants which are subwords of the given output 
examples.
We can then verify in polynomial-time whether the guessed constants form 
a \stw{} $\trans$ which is consistent with the sample $\sample$.

To prove $\NP$-hardness, we consider a formula $\varphi$, 
instance of the one-in-three positive SAT.
The formula $\varphi$ has no negated variables, and is satisfiable if there 
exists an assignment to the boolean variables such that for each clause 
of $\varphi$, exactly one variable evaluates to true.

Formally, let $\mathbb{X}$ be a set of variables and 
let $\varphi \equiv \cl_1 \land \dots \land \cl_n$
such that for every $i \in \set{1,\dots,n}$,  
$\cl_i \equiv \oneinthree(x^i_1,x^i_2,x^i_3)$
with $x^i_1,x^i_2,x^i_3 \in \mathbb{X}$.

Let $\Sigma = \ranked{\lf}{0} \cup \set{\ranked{x}{1}\ |\ x \in \mathbb{X}}$.
Let $\Output = \set{ a, \# }$.
Then, for every clause $\oneinthree(x^i_1,x^i_2,x^i_3)$,
we define $\sample(x^i_1(x^i_2(x^i_3(\lf)))) = a \#$.
Finally, we define $\sample(\lf) = \#$.

We now prove the following equivalence.
There exists a \stw{} $\trans$ such that for all 
$(t,w) \in \sample$, $\trans(t) = w$ if and only if 
$\varphi$ is satisfiable.

$(\Rightarrow)$
Let $\trans = (\Alphabet, \Output, \delta)$ be a \stw{} such that for all 
$(t,w) \in \sample$, $\trans(t) = w$
By definition of $\sample$, we know $\trans(\lf) = \#$ and
and for all $i \in \set{1,\dots,n}$, 
$\trans(x^i_1(x^i_2(x^i_3(\lf)))) = a \#$.

Moreover, if for all $x \in \Var$, we denote 
$\delta(x) = (\before{x}, \after{x})$, with 
$\before{x}, \after{x} \in \Output^*$.
Then, by definition of $\trans$, we have, for $i \in \set{1,\dots,n}$:
\begin{flalign*}
  \trans(x^i_1(x^i_2(x^i_3(\lf)))) &= 
  \before{x^i_1} 
  \before{x^i_2} 
  \before{x^i_3}
  \trans(\lf)
  \after{x^i_3}
  \after{x^i_2}
  \after{x^i_1} \\&= 
  \before{x^i_1} 
  \before{x^i_2} 
  \before{x^i_3}
  \#
  \after{x^i_3}
  \after{x^i_2}
  \after{x^i_1}
\end{flalign*}
We deduce that 
  $\before{x^i_1} 
  \before{x^i_2} 
  \before{x^i_3} = a$ and 
  $\after{x^i_3}
  \after{x^i_2}
  \after{x^i_1} = \varepsilon$.
  
Thus, exactly one of 
$\before{x^i_1}$, 
$\before{x^i_2}$,
$\before{x^i_3}$ must be equal to $a \in \Output$, 
while the other two must be equal to
$\varepsilon$.
Then, $\varphi$ is satisfiable using the boolean assignment that maps 
a variable $x \in \Var$ to $\top$ if $\before{x} = a$, and 
to $\bot$ if $\before{x} = \varepsilon$.

$(\Leftarrow)$
Conversely, assume there exists a satisfying assignment 
$\mu: \Var \to \set{\bot,\top}$ for $\varphi$.
Then, we define the \stw{} $\trans = (\Alphabet, \Output, \delta)$ where 
$\delta(\lf) = \#$ and for all $x \in \Var$
$\delta(x) = (a,\varepsilon)$ if $\mu(x) = \top$, 
and $\delta(x) = (\varepsilon,\varepsilon)$ if $\mu(x) = \bot$.
We then have $\trans(t) = w$ for all $(t,w) \in \sample$.
\end{proof}
\begin{remark}
The NP-completeness proof of~\cite{laurence_phd_2014} could not apply here, because the transducers are more general. Namely, they are allowed to have multiple states in their setting.
\end{remark}

\section{Solving Sequential Word Equation in Polynomial Time}

\rplemma*

\begin{proof}
By definition of \regular{}, $\form$ can be written as 
$\form_1 \land \dots \land \form_l$ for some $l \in \Nat$,
where for $i \neq j$, $\form_i$ and $\form_j$ do not have variables
in common. We can thus check for satisfiability of $\form$ by checking
satisfiability of each $\form_i$ independently.
Let $\psi$ be one of $\form_i$ for $i \in \set{1,\dots,l}$.

By definition of \regular{}, we know there exists $n \in \Nat$ with
$\psi \equiv \y 1 = w_1\, \land \dots \land\, \y n = w_n$, 
and there exist $k \in \Nat$ and 
$\x 0,\dots,\x k \in \Var$, such that 
for all $i \in \set{1,\dots,n}$, 
$w_i \in \Output^*$, and 
$\restrict{\rp_i}{\Var} = \x 0 \cdots \x k$.

The outline of the proof is the following.
For $i \in \set{1,\dots,n}$, we build an 
acyclic DFA, 
denoted $A_i = \getautomaton{\rp_i}{w_i}$,
that recognizes the set
\begin{flalign*}
\{
  \assignment(\x 0) &\spec \assignment(\x 1) \cdots \spec \assignment(\x k)\ |\ 
  \assignment: \Var \rightarrow \Output^* \land
  \assignment(\rp_i) = w_i
\}
\end{flalign*}
where $\spec$ is a special character we introduce, used as a separator.

Then, there exists an assignment
$\assignment: \Var \rightarrow \Output^*$ such that 
for all $i \in \set{1,\dots,n}$,
$\assignment(\rp_i) = w_i$ if and only if 
$A_1 \cap \dots \cap A_n \neq \emptyset$.
We then show that the emptiness of this intersection can be checked
in polynomial time, due to the particular form of the automata.
(In general, checking the emptiness of the intersection of $n$ automata
is a $\PSPACE$-complete problem, and thus takes exponential time to 
check.)

We now give the formal details of the proof.
Let $i \in \set{1,\dots,n}$, 
and $\rp_i = \x 0 u_1 \x 1 \cdots u_k \x k$.
We define 
$A_i = \getautomaton{\rp_i}{w_i} = (Q_i,q_i,\delta_i)$ as follows:
\begin{itemize}
\item 
    $Q_i = \set{0,\dots,k} \times \set{0,\dots,|\rp_i|}$ is the set of states,
\item 
    $q_i = (0,0)$ is the initial state,
\item 
    for $a \in \set{0,\dots,k}$, $b \in \set{0,\dots,|\rp_i|}$,
    \begin{itemize}
    \item 
        $\delta((a,b), \#) = (a+1,b+|u_{a+1}|)$ \\
        if $a < k$ and 
        $\substring{\rp_i}{b}{|u_{a+1}|} = u_{a+1}$,
    \item $\delta((a,b), \sigma) = (a,b+1)$ if 
        the $(b+1)$th letter of $|\rp_i|$ is $\sigma$.
    \end{itemize}
\end{itemize}

We now prove that the intersection 
$A_1 \cap \dots \cap A_n$ can be represented by 
an automaton which has as many states as the smallest $A_i$.
We first compute the intersection between $A_1$ and $A_2$, and show 
that the resulting automaton can be obtained from $A_1$ by deleting 
transitions (see Figure~\ref{figure:intersectionexample}).

We denote the states of $A_1$ by 
  $P = \set{\mkstate{p}{i}{j}\ |\ i \in \set{0,\dots,k}, j \in \set{0,\dots,|w_1|}}$
 and the states of $A_2$ by 
 $Q = \set{\mkstate{q}{i}{j}\ |\ i \in \set{0,\dots,k}, 
  j \in \set{0,\dots,|w_2|}}$.

Let $\rp_1 = \x 0  u_1  \x 1 \cdots u_k  \x k$,
and $\rp_2 = \x 0  v_1  \x 1 \cdots v_k  \x k$,
where $u_1,\dots,u_k,v_1,\dots,v_k \in \Output^*$

We know that whenever there is a transition from 
a state $\mkstate{p}{i}{j}$ to $\mkstate{p}{i'}{j'}$ in $A_1$ then either:
\begin{itemize}
  \item $i' = i$ and $j' = j+1$ ($\Output$-transitions), or 
  \item $i' = i+1$ and $j' = j+|u_{i+1}|$ ($\spec$-transitions)
\end{itemize}
The same property holds for $A_2$, by replacing $u_{i+1}$ with 
$v_{i+1}$.

We compute the cartesian product $B_2$ of $A_1$ and $A_2$.
The states of $B_2$ are pairs from $P \times Q$.
Consider such a state 
$(\mkstate{p}{i}{j}, \mkstate{q}{i'}{j'})$
which is reachable in $B_2$
from the initial state $(\mkstate{p}{0}{0}, \mkstate{q}{0}{0})$.

First, we can show that $i = i'$.
The only transitions that increase $i$ and $i'$ in $A_1$ and $A_2$
are $\spec$-transitions, and they increase $i$ and $i'$ by $1$.
Thus, in the cartesian product $B_2$, we always have $i = i'$.

Similarly, the following invariant holds for the reachable states 
$(\mkstate{p}{i}{j}, \mkstate{q}{i}{j'})$ of $B_2$:
\[
  j-j' = \sum_{k = 1}^i |u_k| - |v_k|
\]

In particular, this means that each state $p \in P$ can be paired 
with at most one state $q \in Q$ in $B_2$ (and each state of $Q$
can be paired with at most one state of $P$).
This entails that $B_2$ can be obtained from $A_1$ or $A_2$ by 
erasing transition, and that it has at most as many reachable states as 
$min(|P|,|Q|)$.

For $3 \leq i \leq n$, we then compute $B_i = A_i \cap B_{i-1}$
similarly, and obtain that $B_n = A_1 \cap \dots \cap A_n$ has at
most as many reachable states as the smallest $A_i$.

\end{proof}

\section{Test Sets for Linear Context-Free Grammars}

\cftestset*

\begin{proof}

Before building the test set, we introduce some notation.

\subparagraph{Graph of $\gram$.}

Define the labeled graph $\getgraph(\gram) = (V,E)$ where 
$V = \NTerm \cup \set{\final}$,
$\final$ is a new symbol, and
$E \subseteq V \times \Prod \times V$ such that:
\begin{itemize}
\item 
  for non-terminals $A,B \in \NTerm$ and a rule $\rul \in \Prod$,  let
  $(A,\rul,B) \in E$ iff
  $\rul$ is of the form $\mkrule{A}{u B v}$ where $u,v \in \Sigma^*$ (i.e., $B$ is the only non-terminal occurring in
  $\rhs$).
\item 
  for a non-terminal $A \in \NTerm$ and $\rul \in \Prod$, 
  $(A,\rul,\final) \in E$ if and only if
  $\rul = \mkrule{A}{\rhs}$ for some $\rhs \in \Sigma^*$.
\end{itemize}

A \emph{path} of $\getgraph(\gram)$ is a (possibly cyclic) sequence of 
edges of $E$, of the form:
$(A_1,\rul_1,A_2) \cdot
(A_2,\rul_2,A_3) 
\cdots
(A_n,\rul_n,A_{n+1})$.
A path is \emph{accepting} if $A_1 = \Start$ and $A_{n+1} = \final$.

\subparagraph{Link between  $\getgraph(\gram)$ and $\gram$.}

Given a rule $\mkrule{A}{u B v} \in \Prod$,
where $A,B \in \NTerm$ and $u,v \in \Term^*$, 
we denote $\west{\rul} = u$ and $\east{\rul} = v$.
For a rule of the form $\mkrule{A}{u}$ where $u \in \Term^*$
we denote $\west{\rul} = u$ and $\east{\rul} = \varepsilon$.
For a path $P = 
(A_1,\rul_1,A_2) \cdot
(A_2,\rul_2,A_3) \cdot
\cdots
(A_n,\rul_n,A_{n+1})$ 
we define $\west{P} = \west{\rul_1} \cdots \west{\rul_n}$,
and $\east{P} = \east{\rul_n} \cdots \east{\rul_1}$.

Each accepting path $P$ in $\getgraph(\gram)$ corresponds to a word 
$\west{P} \cdot \east{P}$ in $\gram$, and
conversely, for any word $w \in \gram$, there exists an
accepting path (not necessarily unique) in $\getgraph(\gram)$
corresponding to $w$.

\subparagraph{Total order on paths.}

We fix an arbitrary total order $<$ on $\Prod$, and extend it
to sequence of edges in $\Prod^*$ as follows.
Given paths $P_1,P_2 \in \Prod^*$, we have 
$P_1 < P_2$ iff 
\begin{itemize}
\item $|P_1| < |P_2|$ (length of $P_1$ is smaller than length of $P_2$), or
\item $|P_1| = |P_2|$ and 
  $P_1$ is smaller lexicographically than $P_2$.
\end{itemize}

A path $P$ is called \emph{optimal} if it is the minimal path from the first
vertex of $P$ to the last vertex of $P$.

\subparagraph{Test set for $\gram$.}

\begin{figure}
    \centering

    \begin{tikzpicture}[
        pnt/.style={->,circle, fill=black, inner sep=0pt, minimum size=4pt},
        arr/.style={->,dashed, thick},
        edd/.style={->},
        inv/.style={draw=none,opacity=0},
        x=0.9cm,
        scale=1
    ]
    
    \node[pnt] at (0,0) (B0) { };
    \node[pnt] at (2,0) (A1) {  };
    \node[pnt] at (3,0) (B1) {  };
    \node[pnt] at (5,0) (A2) {  };
    \node[pnt] at (6,0) (B2) {  };
    \node[pnt] at (8,0) (A3) {  };
    \node[pnt] at (9,0) (B3) {  };
    \node[pnt] at (11,0) (A4) {  };
    \node[pnt] at (12,0) (B4) {  };
    \node[pnt] at (14,0) (A5) {  };
    
    \path[arr] (B0) edge[bend left=65] node[above] { $Q_1$ } (B1);
    \path[arr] (B1) edge[bend left=65] node[above] { $Q_2$ } (B2);
    \path[arr] (B2) edge[bend left=65] node[above] { $Q_3$ } (B3);
    \path[arr] (B3) edge[bend left=35] node[above] { $Q_4$ } (A5);
    
    \path[arr] (B0) edge[bend left=0] node[below] { $P_1$ } (A1);
    \path[arr] (B1) edge[bend left=0] node[below] { $P_2$ } (A2);
    \path[arr] (B2) edge[bend left=0] node[below] { $P_3$ } (A3);
    \path[arr] (B3) edge[bend left=0] node[below] { $P_4$ } (A4);
    \path[arr] (B4) edge[bend left=0] node[below] { $W_5$ } (A5);

    \path[edd] (A1) edge[bend left=0] node[below] { $e_1$ } (B1);
    \path[edd] (A2) edge[bend left=0] node[below] { $e_2$ } (B2);
    \path[edd] (A3) edge[bend left=0] node[below] { $e_3$ } (B3);
    \path[edd] (A4) edge[bend left=0] node[below] { $e_4$ } (B4);
    
    \node[below] at (B0) { $\Start$ };
    \node[below] at (A5) { $\final$ };
    
    \end{tikzpicture}
    \caption{
        The four optimal subpaths $Q_1$, $Q_2$, $Q_3$, and $Q_4$ define
        $15$ alternative paths from $\Start$ to $\final$ which are all
        strictly smaller (with respect to order $<$) than
        $P_1 e_1 P_2 e_2 P_3 e_3 P_4 e_4 W_5$.
    }
    \label{figure:cftestset}
\end{figure}
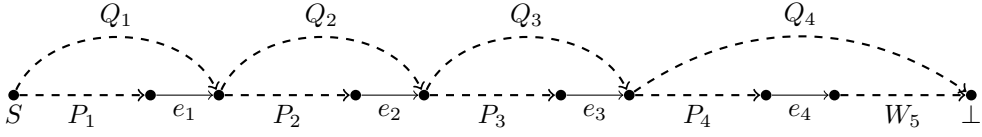

Let $\optimalset_k(\gram)$ be the set of words of $\gram$ corresponding to accepting
paths of the form
$P_1 e_1 P_2 \cdots P_n e_n P_{n+1}$, $n \leq k$,
with $P_i \in \Prod^*$, $e_i \in \Prod$, and 
where for $i \in \set{1,\dots,n+1}$, 
$P_i$ is optimal,
and for $i \in \set{1,\dots,n}$,
$P_i e_i$ is not optimal.
By construction, a path in $\optimalset_k(\gram)$ is uniquely
determined (when it exists) by the choice 
of edges $e_1,\dots,e_n$, as optimal paths between
two vertices are unique.
Therefore, $\optimalset_k(\gram)$
contains at most $\sum_{i = 0}^k |\Prod|^i \leq 
2|\Prod|^k$ words.

We now show that $\optimalset_3(\gram)$ is a test set for $\gram$
(which gives us the desired bound of the theorem: $2|\Prod|^k$).
Assume there exist two morphisms $f, g: \Term^* \to \Gamma^*$ such that 
$\restrict{f}{\optimalset_3(\gram)} = \restrict{g}{\optimalset_3(\gram)}$ and 
there exists $w \in G$ such that $f(w) \neq g(w)$.

By assumption, $w$ does not belong to $\optimalset_3(\gram)$, and must correspond 
to a path $P = P_1 e_1 P_2 \cdots P_n e_n P_{n+1}$ for $n \geq 4$, such that
for $i \in \set{1,\dots,n+1}$, 
$P_i$ is optimal, and $P_i e_i$ is not optimal.
We pick $w$ having the property $f(w) \neq g(w)$ such that the path $P$ is the 
smallest possible (according to the order $<$ defined above).

The path $P$ can be 
written $P_1 e_1 P_2 e_2 P_3 e_3 P_4 e_4 W_5$ where 
for $i \in \set{1,2,3,4}$, $P_i$ is optimal, and 
$P_i e_i$ is not optimal ($W_5$ is not necessarily optimal).
For $i \in \set{1,2,3}$,
we define $Q_i$ to be the optimal path from the source of
$P_i e_i$ to its target; hence $Q_i < P_i e_i$.
Moreover, $Q_4$ is defined to be the optimal path
from the source of $P_4 e_4 W_5$ to its target, with $Q_4 < P_4 e_4 W_5$.
Effectively, as shown in Figure~\ref{figure:cftestset},
this defines $15$ paths that can be derived from $P$
by replacing subpaths by their corresponding optimal path 
($Q_1$, $Q_2$, $Q_3$, $Q_4$).

Let $P'$ be one of those $15$ paths (where at least one subpath 
has been replaced by its optimal counterpart
$Q_1$, $Q_2$, $Q_3$, or $Q_4$), and let $w' \in \gram$
be the word corresponding to $P'$.
By construction of $P'$, and by definition of the order $<$, 
we have $P' < P$.
Since we have chosen $P$ to be the smallest possible path such that $f$ and $g$
are not equal on the corresponding word, we deduce that
$f(w') = g(w')$.

To conclude, 
we show that we obtain a contradiction, thanks to 
Lemma~\ref{lemma:t4l4}.
For this, 
we construct two morphisms $f', g': \Sigma_4 \to \Gamma$ as follows
($i$ ranges over $\set{1,2,3,4}$ and $j$ over $\set{1,2,3}$):
\begin{itemize}
\item $f'(a_i) = f(\west{Q_i})$,
\item $f'(\closing{a_i}) = f(\east{Q_i})$,
\item $f'(b_j) = f(\west{P_j e_j})$,
\item $f'(\closing{b_j}) = f(\east{P_j e_j})$.
\item $f'(b_4) = f(\west{P_4 e_4 W_5})$,
\item $f'(\closing{b_4}) = f(\east{P_4 e_4 W_5})$.
\end{itemize}
The morphism $g'$ is defined similarly, using $g$ instead of $f$.
We can then verify that $f'$ and $g'$
coincide on $T_4$, but are not equal on the word 
$b_4 \,b_3\,b_2\,b_1\,
    \closing{b_1}\,\closing{b_2}\,\closing{b_3}\,\closing{b_4}
    \in L_4$, thus contradicting Lemma~\ref{lemma:t4l4}.
\end{proof}

\section{Lower Bound Proof for the Tree Test Sets}\label{proof:lowerbound}

\lowerbound*

\begin{proof}
Our proof is inspired by the lower bound proof for test sets of context-free
languages~\cite{plandowski_testset_1994,plandowski_testset_1995}. 
However, that lower bound did not work for context-free grammars with the 
extra assumption that all non-terminal symbols are starting symbols.
Therefore, their proof cannot be applied for domains where all states are 
initial.
Our contribution is a variant which shows that, 
even when all states of the domain are initial, 
the lower bound still holds and the minimal test-set has a cubic size.

For $n \geq 1$, we first define the domain 
$\domain_{n} = (\Sigma,\State,I=\State,\delta)$ containing
linear trees (lists) of depth $1$ to $3$, and using $n$ different symbols of each
level.
Formally, we have
(we use a functional notation for $\delta$, as $\delta$ is
here deterministic):
\begin{itemize}
\item 
    $\Sigma = \set{  
        \ranked{\noda_j}{1}, \ranked{\nodb_j}{1},
        \ranked{\leaf_j}{0}\ 
            |\ 1 \leq j \leq n}$,
\item $\State = \set{ q_2, q_1, q_0 }$,
\item $\delta(\ranked{\noda_j}{1}, q_2) = (q_1)$,
\item $\delta(\ranked{\nodb_j}{1}, q_1) = (q_0)$,
\item $\delta(\ranked{\leaf_j}{0}, q_0) = ()$.
\end{itemize}

$\domain_n$ recognizes $n^3+n^2+n$ trees.
Our goal is to prove, by contradiction, that $\domain_n$ does not have a 
\treetestset{} $T \subset \domain_n$ of size less than $n^3$.
Let $t = \noda_x(\nodb_y(\leaf_z)) \in \domain_n \setminus T$ for some arbitrary $x, y, z \in [1, n]$. $t$ exists when the size of $T$ is strictly less than $n^3$.
We construct two \stw{s} $\trans_1$ and $\trans_2$ such 
that
$\restrict{\semantics{\trans_1}}{T} = 
  \restrict{\semantics{\trans_2}}{T}$ but 
$
    \semantics{\trans_1}(t) \neq  
    \semantics{\trans_2}(t)$, contradicting the fact that 
    $T$ is a \treetestset.
    
Let $\Output = \set{p,q}$ be an alphabet, and 
$\trans_1 = (\Sigma, \Output, \delta_1)$,
$\trans_2 = (\Sigma, \Output, \delta_2)$, where
\begin{itemize}
\item $\delta_1(\noda_j) = (\varepsilon,pq)$ if $j = x$,
\item $\delta_1(\noda_j) = (\varepsilon,\varepsilon)$ otherwise,
\item $\delta_1(\nodb_j) = (\varepsilon,\varepsilon)$ if $j = y$,
\item $\delta_1(\nodb_j) = (p,q)$ otherwise,
\item $\delta_1(\leaf_j) = (qp)$ if $j = z$,
\item $\delta_1(\leaf_j) = (\varepsilon)$ otherwise,
\end{itemize}
and
\begin{itemize}
\item $\delta_2(\noda_j) = (pq,\varepsilon)$ if $j = x$,
\item $\delta_2(\noda_j) = (\varepsilon,\varepsilon)$ otherwise,
\item $\delta_2(\nodb_j) = (\varepsilon,\varepsilon)$ if $j = y$,
\item $\delta_2(\nodb_j) = (p,q)$ otherwise,
\item $\delta_2(\leaf_j) = (qp)$ if $j = z$,
\item $\delta_2(\leaf_j) = (\varepsilon)$ otherwise.
\end{itemize}
    
We can verify that 
$\restrict{\semantics{\trans_1}}{T} = 
  \restrict{\semantics{\trans_2}}{T}$, but
  $\trans_1(t) \neq  
    \trans_2(t)$, as
 $\trans_1(t) = pqqp$ and 
 $\trans_2(t) = qppq$.
 
We conclude that the only \treetestset{} of $\domain_n$ is $\domain_n$ itself,
which contains $n^3 = (\frac{|\Sigma|}{3})^3$ words. Moreover, the (syntactic)
size of $\domain_n$ is $O(n)$.
\end{proof}

We note that the above transducers $\delta_1$ and $\delta_2$ have the same output on all  $\leaf_k$ and on all  $\nodb_j(\leaf_k)$. Therefore, even if we interactively ask questions as for Theorem~\ref{theorem:linearinteraction}, these questions will not be able to resolve the ambiguity which will appear only at the specific $\noda_x(\nodb_y(\leaf_z))$.

\section{Construction of $\optimalset_3(\gram)$}

To construct $\optimalset_3(\gram)$ for a linear context-free 
grammar $\gram = (\NTerm,\Term,\Prod,\Start)$,
we precompute in time
$O(|\NTerm|^2 |\Prod|)$, for each pair of vertices 
$(A,B)$, the optimal path from $A$ to $B$ in 
$\getgraph(\gram)$.
Then for each possible choice of at most $3$ edges 
$e_1 = (A_1,r_1,B_1)$,
\dots
$e_n = (A_n,r_n,B_n)$,
with $0 \leq n \leq 3$, 
we construct the path 
$P = P_1 e_1 \dots P_n e_n P_{n+1}$ where 
each $P_i$ is the optimal path from $A_{i-1}$ to $B_i$
(if it exists) with $A_0 = \Start$ and $B_{n+1} = \final$
by convention.
We then add 
the word corresponding to $P$ to our result.

To conclude, since the length of each optimal path is bounded by
$|\NTerm|$, we can construct 
$\optimalset_3(\gram)$ in time $O(|\NTerm|  \cdot |\Prod|^3)$.

\section{Proof of Running Time of Algorithm~\ref{algo3}}

\linearinteraction*

\begin{proof}

The correctness and the polynomial running time of Algorithm~\ref{algo3} can
be proved similarly to Algorithm~\ref{algo2}.

Note that we can check whether the NFA $A$ recognizes only one word using the 
following polynomial time procedure.
First, check if there exists a word $w$ recognized by $A$.
If there is, pick a minimal word $w \in \A$, and 
compute the automaton $A \cap B$, where $B$ recognizes all words different than 
$w$ (the size of $B$ is roughly $|w|$).
If the automaton $A \cap B$ is empty, then $A$ recognizes only $w$,
otherwise $A$ recognizes more than one word.

The crucial part of Algorithm~\ref{algo3} is that it invokes the oracle 
$\utrans$ at most $O(|\domain|)$ times. More precisely, we show that 
Algorithm~\ref{algo3} invokes $\utrans$ at most 
$|Q| + 3 \sum_{(q,\ranked{f}{k},(q_1,\dots,q_k) \in \delta} 1 + k$ times,
which is $|Q| + 3|\domain|$, and in $O(|\domain|)$.

Let $\invoked \subseteq T$ be the set of trees for which 
Algorithm~\ref{algo3} invokes the oracle.
Let $\tmin \subseteq T$ be the set of trees which are of the 
form $\gettree(w_\nterm)$ for some $\nterm \in \gram$
(where $w_\nterm$ is the minimal word that can be produced from 
$\nterm$, see Lemma~\ref{lemma:trees}).
Note that, by construction of $\gram$, $|\tmin| \leq |Q|$.

Consider the set $\invoked' = \invoked \setminus \tmin$.
Let $\winvoked' = \set{\default{\Sigma}(t)\ |\ t \in \invoked'}$
(or equivalently, $\invoked' = \set{\gettree(w) \ |\ w \in \winvoked'}$).
We want to prove that $|\invoked'| \leq 3|\domain|$, thus implying that 
$|\invoked| \leq 3|\domain| + |Q|$.
Assume by contradiction $|\invoked'| > 3|\domain|$.

Remember that, by construction of $T$, we have 
$\winvoked' \subseteq \Lin{\gram}$.
For $w \in \winvoked$, consider the first non-epsilon rule in some derivation 
of $w$ in $\Lin{\gram}$. By construction, $\Lin{\gram}$ has at most 
$|\domain|$ rules.

Moreover, since $\winvoked$ contains strictly more than $3|\domain|$ words,
$\winvoked$ must contain at least four words that share the same first 
non-epsilon rule in their derivation.
Let $w_a,w_b,w_c,w_d$ be four such words, and 
$t_a,t_b,t_c,t_d$ their corresponding trees 
(with $t_l = \gettree(w_l)$ for $l \in \set{a,b,c,d}$).

Without loss of generality, assume that Algorithm~\ref{algo3} invoked 
$\utrans$ on the trees $t_a$, $t_b$, $t_c$, and $t_d$ in that order.
By construction of $\Lin{\gram}$, and from the fact that 
$w_a$, $w_b$, $w_c$, and $w_d$ share the first non-epsilon rule, 
we know there exists 
$\ranked{f}{k} \in \Sigma$, $i \in \set{1,\dots,k}$, 
$t_1,\dots,t_{i-1},t_{i+1},\dots,t_k \in T$, and 
$t_i^a,t_i^b,t_i^c,t_i^d \in T$
such that:
\begin{flalign*}
t_a = &f(t_1,\dots,t_{i-1},t_i^a,t_{i+1},\dots,t_k)\\
t_b = &f(t_1,\dots,t_{i-1},t_i^b,t_{i+1},\dots,t_k)\\
t_c = &f(t_1,\dots,t_{i-1},t_i^c,t_{i+1},\dots,t_k)\\
t_d = &f(t_1,\dots,t_{i-1},t_i^d,t_{i+1},\dots,t_k).
\end{flalign*}

Said otherwise $t_a$, $t_b$, $t_c$, and $t_d$ only differ on their $i$th
subtree.
Let $x_i^l = \utrans(t_i^l)$
and $x_l = \utrans(t_l)$, for $l \in \set{a,b,c,d}$.

Given an assignment $\morph: \SigmaPairs \to \Output^*$, 
we define $u_\morph = \morph(f,0)\cdot \utrans(t_1) \cdot \morph(f,1) \cdots \morph(f,i-1)$, and $v_\morph =  \morph(f,i+1) \cdots \utrans(t_k)\cdot \morph(f,k)$.

We consider several cases, all of them leading to a contradiction.
\begin{itemize}
\item 
    There exist $x_i^k,x_i^l\in\{x_i^a, x_i^b, x_i^c\}$ such that
    $x_1$ is not a prefix of $x_2$ and 
    $x_2$ is not a prefix of $x_1$.
    Assume without loss of generality $x_1 = x_i^a$ and $x_2 = x_i^b$.
    Since 
    Algorithm~\ref{algo3} invoked the oracle
    $\utrans$ on the tree $t_d$, there are two
    assignments $\morph,\morph': \SigmaPairs \to \Output^*$ such 
    that $u_\morph x_i^d v_\morph \in A$,
    $u_{\morph'} x_i^d v_{\morph'} \in A$,
    and
    $u_\morph x_i^d v_\morph \neq u_{\morph'} x_i^d v_{\morph'}$.
    Using the invariant of Algorithm~\ref{algo3} on trees $t_a$ and $t_b$, 
    we know:
\begin{align*}
x_a = u_\morph x_i^a v_\morph &= u_{\morph'} x_i^a v_{\morph'}\\
x_b = u_\morph x_i^b v_\morph &= u_{\morph'} x_i^b v_{\morph'}
\end{align*}
    Finally, since $x_i^a$ is not a prefix of $x_i^b$, and
    $x_i^b$ is not a prefix of $x_i^a$, we have
        $\lcp(x_a,x_b) = u_\morph\, \lcp(x_i^a,x_i^b)$ and
        $\lcp(x_a,x_b) = u_{\morph'}\, \lcp(x_i^a,x_i^b)$.
        Thus, 
        $u_{\morph} = u_{\morph'}$ and we can deduce
        $v_{\morph} = v_{\morph'}$.
        Thus, 
        $u_\morph\, x_i^d\, v_\morph = u_{\morph'}\, x_i^d\, v_{\morph'}$, and
        we have a contradiction.
\item
    Two elements of $\{x_i^a, x_i^b, x_i^c\}$ are equal.
    For instance if $x_i^a = x_i^b$, then 
    Algorithm~\ref{algo3} could not have invoked the oracle $\utrans$
    on the tree $t_b$, as the only possible solution for 
    $x_b$ is $x_a$.
\item 
    If we are not in one of the previous cases, we know that
    $x_i^a, x_i^b, x_i^c$ are of the 
    form $x_1$, $x_1 x_2$, $x_1 x_2 x_3$ (not necessarily in that order)
    for some $x_1,x_2,x_3 \in \Output^*$,
    and $x_2 \neq \varepsilon$ and $x_3 \neq \varepsilon$.
    Consider the case where $x_2 x_3 \neq x_3 x_2$.
            
    Since Algorithm~\ref{algo3} invoked $\utrans$ on the tree $t_d$, 
    there are two assignments $\morph,\morph': \SigmaPairs \to \Output^*$ such 
    that $u_\morph x_i^d v_\morph \in A$,
    $u_{\morph'} x_i^d v_{\morph'} \in A$,
    and $u_\morph x_i^d v_\morph \neq u_{\morph'} x_i^d v_{\morph'}$.
    Using the invariant of Algorithm~\ref{algo3}, we know:
\begin{align*}
u_\morph x_1 v_\morph &= u_{\morph'}  x_1  v_{\morph'}\\
u_\morph x_1 x_2 v_\morph &= u_{\morph'}  x_1 x_2  v_{\morph'}\\
u_\morph x_1 x_2 x_3 v_\morph &= u_{\morph'} x_1 x_2 x_3 v_{\morph'}
\end{align*}
    Without loss of generality, assume that $u_{\morph'}$ is a prefix of 
    $u_\morph$ and $v_\morph$ is a suffix of $v_{\morph'}$.
    So $u_\morph = u_{\morph'} u''$ and 
    $v_{\morph'} = v'' v_\morph$ for some $u'',v'' \in \Output^*$.
    Then, we have
    $u'' \, x_1 = x_1 \, v''$,
    $u'' \, x_1 x_2 = x_1 x_2 \, v''$, and 
    $u'' \, x_1 x_2 x_3 = x_1 x_2 x_3 \, v''$.
    We deduce,
    $x_1 \, v'' x_2 = x_1 x_2 \, v''$, and 
    $v''$ commutes with $x_2$.
    
    Similarly, $v''$ commutes with $x_2x_3$.
    Assume by contradiction that $v'' \neq \varepsilon$.
    Then, by Lemma~\ref{lemma:commute}, 
    $x_2$ must commute with $x_2x_3$, and 
    $x_2 x_2 x_3 = x_2 x_3 x_2$, which implies 
    $x_2 x_3 = x_3 x_2$, and $x_2$ commutes with $x_3$,
    contradicting our assumption.
    This means that $v'' = \varepsilon$, 
    $v_\morph = v_{\morph'}$, and 
    $u_\morph  = u_{\morph'}$.
    We thus conclude 
    $u_\morph x_i^d v_\morph = u_{\morph'} x_i^d v_{\morph'}$,
    contradicting the fact that Algorithm~\ref{algo3} invoked 
    $\utrans$ on tree $t_d$.
\item  
    The last case is when the $x_i^a, x_i^b, x_i^c$ are of the 
    form $x_1$, $x_1 x_2$, $x_1 x_2 x_3$ (not necessarily in that order)
    for some $x_1,x_2,x_3 \in \Output^*$,
    with $x_2 \neq \varepsilon$ and $x_3 \neq \varepsilon$,
    and $x_2 x_3 = x_3 x_2$.
    
    Since Algorithm~\ref{algo3} invoked $\utrans$ on the tree $t_c$, 
    there are two assignments $\morph,\morph': \SigmaPairs \to \Output^*$ such 
    that $u_\morph x_i^c v_\morph \in A$,
    $u_{\morph'} x_i^c v_{\morph'} \in A$, 
    and $u_\morph x_i^c v_\morph \neq u_{\morph'} x_i^c v_{\morph'}$,
    where $A$ is the automaton
    constructed in Algorithm~\ref{algo3} at the iteration where 
    $\utrans(t_c)$ was invoked.

    Without loss of generality, assume that $u_{\morph'}$ is a prefix of 
    $u_\morph$ and $v_\morph$ is a suffix of $v_{\morph'}$.
    So $u_\morph = u_{\morph'} u''$ and 
    $v_{\morph'} = v'' v_\morph$ for some $u'',v'' \in \Output^*$.
    We consider three subcases:
    \begin{itemize}
    \item Case $x_i^c = x_1 x_2 x_3$.
    Using the invariant of Algorithm~\ref{algo3} for trees $t_a$ and $t_b$, 
    we know:
\begin{align*}
u_\morph x_1 v_\morph &= u_{\morph'}  x_1  v_{\morph'}\\
u_\morph x_1 x_2 v_\morph &= u_{\morph'}  x_1 x_2  v_{\morph'}
\end{align*}
    Then, we have
    $u'' \, x_1 = x_1 \, v''$,
    $u'' \, x_1 x_2 = x_1 x_2 \, v''$.
    We deduce,
    $x_1 \, v'' x_2 = x_1 x_2 \, v''$, and 
    $v''$ commutes with $x_2$.
    
    Thus, $x_2$ commutes both with $v''$ and $x_3$.
    By Lemma~\ref{lemma:commute}, $v''$ commutes with $x_3$.
    We deduce that 
    $x_1 x_2 x_3 v'' = x_1 v'' x_2 x_3 = u'' x_1 x_2 x_3$, 
    and finally that 
    $u_{\morph} x_1 x_2 x_3 v_{\morph} = u_{\morph'} x_1 x_2 x_3 v_{\morph'}$,
    contradicting $u_\morph x_i^c v_\morph \neq u_{\morph'} x_i^c v_{\morph'}$.
    
    \item Case $x_i^c = x_1 x_2$.
    Using the invariant of Algorithm~\ref{algo3} for the trees $t_a$ and $t_b$, 
    we know:
\begin{align*}
u_\morph x_1 v_\morph &= u_{\morph'}  x_1  v_{\morph'}\\
u_\morph x_1 x_2 x_3 v_\morph &= u_{\morph'}  x_1 x_2 x_3 v_{\morph'}
\end{align*}
    Then, we have
    $u'' \, x_1 = x_1 \, v''$,
    $u'' \, x_1 x_2 x_3 = x_1 x_2 x_3 \, v''$.
    We deduce,
    $x_1 \, v'' x_2 x_3 = x_1 x_2 x_3 \, v''$, and 
    $v''$ commutes with $x_2 x_3$.

    Since $x_2 x_3$ commutes both with $v''$ and $x_2$ (as $x_2$ and 
    $x_3$ commute), we know by Lemma~\ref{lemma:commute} that 
    $x_2$ and $v''$ commute.
    We deduce that 
    $x_1 x_2 v'' = x_1 v'' x_2 = u'' x_1 x_2$, 
    and finally that 
    $u_{\morph} x_1 x_2 v_{\morph} = u_{\morph'} x_1 x_2 v_{\morph'}$,
    contradicting $u_\morph x_i^c v_\morph \neq u_{\morph'} x_i^c v_{\morph'}$.
    
    \item Case $x_i^c = x_1$.
    Using the invariant of Algorithm~\ref{algo3} for the trees $t_a$ and $t_b$, 
    we know:
\begin{align*}
u_\morph x_1 x_2 v_\morph &= u_{\morph'}  x_1 x_2 v_{\morph'}\\
u_\morph x_1 x_2 x_3 v_\morph &= u_{\morph'}  x_1 x_2 x_3  v_{\morph'}
\end{align*}
    Then, we have
    $u'' \, x_1 x_2 = x_1 x_2 \, v''$,
    $u'' \, x_1 x_2 x_3 = x_1 x_2 x_3 \, v''$.
    We deduce,
    $x_1 x_2 \, v'' x_3 = x_1 x_2 x_3 \, v''$, and 
    $v''$ commutes with $x_3$.
    
    Thus, $x_3$ commutes both with $v''$ and $x_2$.
    By Lemma~\ref{lemma:commute}, $v''$ commutes with $x_2$.
    Moreover, since 
    $u'' \, x_1 x_2 = x_1 x_2 \, v''$,
    we have 
    $u'' \, x_1 x_2 = x_1 \, v'' x_2$,
    and $u'' \, x_1 = x_1 \, v''$.
    We conclude that 
    $u_{\morph} x_1 v_{\morph} = u_{\morph'} x_1 v_{\morph'}$,
    contradicting $u_\morph x_i^c v_\morph \neq u_{\morph'} x_i^c v_{\morph'}$.
    \end{itemize}
\end{itemize}

\end{proof}

\section{Proof of Theorem~\ref{theorem:withvalues}}

\theoremwithvalues*

\begin{proof}
Let $\domain$ be a domain with the property above.
First, remark that for some $s, u\in \Output^*$, if $t'\in \domain'$ and $t'$ contains $\dstring{s}$, then all the trees obtained from $t'$ by replacing $\dstring{s}$ by $\dstring{u}$ are also in $\domain'$.

We build the linear tree test set $\exts$ as follows.

Let us associate to every symbol $f^{(k)} \in \Sigma$ and to every position in the arity $i \in [1,k]$ the tree $t_i$ as provided by the hypothesis. Let us take one of its $t'_i \in \domain'$ (denoted $\treec{f}{i}{s}$) such that $t'_i$ contains at least one $\dstring{s}$ for some $s\in\Output^*$, the other $\dstring{}$ being mapped to any other constant.
For all other $u\in\Output^*$, remark that $\treec{f}{i}{u}\in \domain'$.

Define $\exts$ containing all ``default'' trees and all their variations, changing one \ta{} to \tb{}:

$$\exts = \bigcup_{f^{(k)}\in \Sigma} \left( \{f^{(k)}(\treec{f}{i}{\ta}_{i\in[1,k]})\} \cup \{ f^{(k)}(\treec{f}{i}{\begin{cases} \tb & \mbox{if } i = j \\ \ta & \mbox{else} \end{cases}}_{i\in[1,k]}) | j\in[1,k]\}\right)$$.

Note that $\exts$ is included in $\domain'$, since the original tree was in $\domain'$.

The size of the set $\exts$ is at most $2|\Sigma|\cdot A$.
We will now prove that $\exts$ is a tree test set for $\domain'$. Indeed,
suppose that we have two transducers $\trans_1$, and $\trans_2$, and that they are equal on $\exts$.
We will show the strong result that they have the \textit{same constants}.

Let $f^{(k)} \in \Sigma$.
Since $\treec{f}{i}{\ta},\treec{f}{i}{\tb} \in \exts$ for $i \in [1, k]$, we can write $$\semantics{\trans_1}(\treec{f}{i}{\ta}) = \semantics{\trans_2}(\treec{f}{i}{\ta})\quad\wedge\quad \semantics{\trans_1}(\treec{f}{i}{\tb}) = \semantics{\trans_2}(\treec{f}{i}{\tb})$$
which, for some $w_i$ and $z_i$ can be rewritten to:
$$w_i \cdot \ta \cdot z_i =  w'_i \cdot \ta \cdot z'_i\quad\wedge\quad
w_i \cdot \tb \cdot z_i =  w'_i \cdot \tb \cdot z'_i$$
We apply Lemma~\ref{lemma:annihilation} to conclude that $w_i = w'_i$ and $z_i = z'_i$, so we can remove the primes.

By hypothesis, we have that:
$$
\semantics{\trans_1}(f^{(k)}(\treec{f}{i}{\ta}_{i\in[1,k]})) = \semantics{\trans_2}(f^{(k)}(\treec{f}{i}{\ta}_{i\in[1,k]})) $$
which we can rewrite to:
$$u_0 w_1 \ta z_1 \cdot u_1 \cdots 
w_k \ta z_k \cdot u_k =
u'_0 w_1 \ta z_1 \cdot u'_1 \cdots 
w_k \ta z_k \cdot u'_k
$$
We also have $k$ other similar equalities by changing any of the \ta{} by a \tb{} at the same place in the two sides of the equation.
Using Lemma~\ref{lemma:annihilation}, and by changing the first \ta{} to \tb{}, we obtain that $u_0 w_1 = u'_0 w_1$ so $u_0 = u'_0$. After simplifying the equations, it remains that:
$$u_1 \cdots 
w_k \ta z_k \cdot u_k =
u'_1 \cdots 
w_k \ta z_k \cdot u'_k
$$
so we can continuously apply the lemma to obtain that $u_1 = u'_1, \ldots u_k = u'_k$.
Hence $\exts$ is a tree test set of linear size.
\end{proof}

\end{document}